\documentclass[11pt]{article}
\usepackage[english]{babel}
\usepackage{cite}
\usepackage{amssymb,amsmath,amsthm,mathrsfs}

\newtheorem{lemma}{Lemma}
\newtheorem{proposition}{Proposition}
\newtheorem{theorem}{Theorem}
\newtheorem{corollary}{Corollary}

\theoremstyle{definition}
\newtheorem{definition}{Definition}

\theoremstyle{remark}
\newtheorem{remark}{Remark}
\newtheorem{example}{Example}

\newcommand{\hh}{\mathcal{H}}
\newcommand{\ptu}{\mathsf{ptu}\hspace{0.2mm}(\hh)}
\newcommand{\cptu}{\mathsf{cptu}\hspace{0.2mm}(\hh)}

\newcommand{\poma}{\mathscr{P}}
\newcommand{\copoma}{\mathscr{K}}

\newcommand{\hbo}{\boldsymbol{h}}
\newcommand{\hon}{h_1}
\newcommand{\htw}{h_2}
\newcommand{\htr}{h_3}
\newcommand{\hj}{h_j}
\newcommand{\kajk}{\varkappa_{jk}}
\newcommand{\kakj}{\varkappa_{kj}}
\newcommand{\kao}{\varkappa_{11}}
\newcommand{\kat}{\varkappa_{22}}
\newcommand{\katr}{\varkappa_{33}}
\newcommand{\kaot}{\varkappa_{12}}
\newcommand{\kaotr}{\varkappa_{13}}
\newcommand{\kattr}{\varkappa_{23}}
\newcommand{\ko}{\kappa_{1}}
\newcommand{\kt}{\kappa_{2}}
\newcommand{\ktr}{\kappa_{3}}

\newcommand{\lcs}{\mathcal{X}}
\newcommand{\bores}{\mathscr{E}}

\newcommand{\clco}{\overline{\mathrm{co}}\hspace{0.3mm}}

\newcommand{\opu}{\hat{U}}
\newcommand{\opua}{\hat{U}^\ast}

\newcommand{\traspo}{\mathfrak{T}\hspace{0.3mm}}
\newcommand{\lino}{\mathfrak{E}}

\newcommand{\oba}{\Psi}
\newcommand{\vjklm}{v_{\hspace{-0.3mm}jklm}(\oba;g)}
\newcommand{\vjkkj}{v_{\hspace{-0.3mm}jkkj}(\oba;g)}

\newcommand{\unihh}{\mathrm{U}\hspace{0.1mm}(\hh)}
\newcommand{\sunig}{\mathrm{SU}\hspace{0.1mm}}
\newcommand{\unig}{\mathrm{U}\hspace{0.1mm}}
\newcommand{\rot}{\mathrm{SO}\hspace{0.1mm}(3)}
\newcommand{\repr}{\mathsf{V}}
\newcommand{\tr}{\mathrm{tr}}

\newcommand{\diag}{\hspace{0.3mm}\mathsf{diag}\hspace{0.2mm}}
\newcommand{\gak}{\gamma_k}
\newcommand{\gao}{\gamma_1}
\newcommand{\gat}{\gamma_2}
\newcommand{\gatr}{\gamma_3}

\newcommand{\trasp}{^{\hspace{-0.2mm}\mbox{\tiny $\mathsf{T}$}}}
\newcommand{\rcub}{\mathbb{R}^3}
\newcommand{\ccub}{\mathbb{C}^3}
\newcommand{\scapro}{\langle\cdot\hspace{0.5mm},\cdot\rangle}

\newcommand{\ranglehs}{\rangle_{\mbox{\tiny \rm HS}}}
\newcommand{\hspro}{\scapro_{\mbox{\tiny \rm HS}}}
\newcommand{\norhs}{\|_{\mbox{\tiny \rm HS}}}

\newcommand{\rangler}{\rangle_{\rcub}\hspace{-0.3mm}}
\newcommand{\ranglec}{\rangle_{\ccub}\hspace{-0.3mm}}
\newcommand{\zzz}{\mathbf{0}}
\newcommand{\spb}{\boldsymbol{\hat{S}}}
\newcommand{\spz}{\hat{S}_0}
\newcommand{\spo}{\hat{S}_1}
\newcommand{\spt}{\hat{S}_2}
\newcommand{\sptr}{\hat{S}_3}
\newcommand{\spj}{\hat{S}_j}
\newcommand{\spk}{\hat{S}_k}
\newcommand{\spl}{\hat{S}_l}
\newcommand{\siz}{\hat{\sigma}_0}
\newcommand{\sio}{\hat{\sigma}_1}
\newcommand{\sit}{\hat{\sigma}_2}
\newcommand{\sitr}{\hat{\sigma}_3}
\newcommand{\aaa}{\boldsymbol{a}}
\newcommand{\bbb}{\boldsymbol{b}}
\newcommand{\rrb}{\boldsymbol{r}}

\newcommand{\ru}{r_1}
\newcommand{\rtr}{r_3}
\newcommand{\az}{a_0}
\newcommand{\ao}{a_1}
\newcommand{\at}{a_2}
\newcommand{\atr}{a_3}
\newcommand{\bz}{b_0}
\newcommand{\mase}{\mathscr{D}(t)}
\newcommand{\mage}{\mathscr{L}}
\newcommand{\rema}{\Lambda}
\newcommand{\remat}{\Lambda\trasp}

\newcommand{\pdm}{\mathfrak{D}}
\newcommand{\pds}{\pdm_t}
\newcommand{\coco}{\hat{J}}

\newcommand{\defi}{\mathrel{\mathop:}=}
\newcommand{\ifed}{=\mathrel{\mathop:}}
\newcommand{\defar}{\overset{\mathrm{def}}{\Longleftrightarrow}}

\newcommand{\dime}{\mathtt{N}}
\newcommand{\quz}{q(\dime)}

\newcommand{\simpl}{\Delta_n}
\newcommand{\simplmo}{\Delta_{n-1}}
\newcommand{\noappro}{\not\approx}
\newcommand{\fun}{\mathscr{F}}

\newcommand{\acca}{\mathsf{h}}

\newcommand{\gi}{\mathsf{g}}
\newcommand{\go}{\gi_1}
\newcommand{\gj}{\gi_j}
\newcommand{\gn}{\gi_n}

\newcommand{\ide}{\mathrm{Id}}
\newcommand{\id}{\hat{I}}

\newcommand{\errep}{\mathbb{R}^{\mbox{\tiny $+$}}\hspace{-0.5mm}}
\newcommand{\erreps}{\mathbb{R}_{\hspace{0.3mm}\ast}^{\mbox{\tiny $+$}}\hspace{-0.5mm}}
\newcommand{\erren}{\mathbb{R}^n}

\newcommand{\entro}{\mathscr{E}}

\newcommand{\pap}{\mathsf{p}}
\newcommand{\norp}{\|_{\pap}}
\newcommand{\norpp}{\|_{[\pap]}}
\newcommand{\norpo}{\|_{[1]}}
\newcommand{\noro}{\|_{1}}
\newcommand{\nori}{\|_{\infty}}

\newcommand{\nep}{\mathscr{A}_{\pap}}
\newcommand{\neinf}{\mathscr{A}_{\infty}}
\newcommand{\alp}{\alpha_{\pap}}
\newcommand{\alinf}{\alpha_{\infty}}
\newcommand{\pkp}{p_{k}^{\hspace{0.3mm}\pap}}

\newcommand{\supro}{\mathfrak{P}}
\newcommand{\suproz}{\mathfrak{P}_0}
\newcommand{\suprop}{\mathfrak{P}^\perp}

\newcommand{\orthp}{\hat{P}}
\newcommand{\orthq}{\hat{Q}}
\newcommand{\orthk}{\hat{P}_k^{\phantom{\ast}}}

\newcommand{\stah}{\mathcal{S}(\hh)}
\newcommand{\hrho}{\hat{\rho}}
\newcommand{\hrhoq}{\hrho^{\hspace{0.4mm}q}}
\newcommand{\hrhot}{\hrho^{\hspace{0.4mm}2}}
\newcommand{\mms}{\hrho_{\star}}
\newcommand{\ho}{\hat{\omega}}

\newcommand{\asy}{\hrho_{\infty}}

\newcommand{\unop}{\hat{U}}
\newcommand{\unopn}{\hat{U}^{\hspace{0.3mm}n}}
\newcommand{\unopa}{\hat{U}^{\ast}}
\newcommand{\unopan}{\hat{U}^{\ast\hspace{0.2mm}n}}

\newcommand{\vne}{\mathscr{S}}
\newcommand{\tseq}{\mathscr{T}_q}

\newcommand{\tseo}{\mathscr{T}_1}
\newcommand{\tset}{\mathscr{T}_2}

\newcommand{\req}{\mathscr{R}_q}
\newcommand{\reqo}{\mathscr{R}_1}
\newcommand{\reqt}{\mathscr{R}_2}
\newcommand{\roq}{\varrho_q}
\newcommand{\roo}{\varrho_1}

\newcommand{\teq}{\tau_q}
\newcommand{\teo}{\tau_1}

\newcommand{\proj}{\hat{\psi}_{\hspace{-0.2mm}j}^{\phantom{q}}}
\newcommand{\prok}{\hat{\psi}_k^{\phantom{q}}}
\newcommand{\vproj}{\hat{\phi}_j^{\phantom{q}}}

\newcommand{\xo}{x_1^{\phantom{q}}}
\newcommand{\xk}{x_k^{\phantom{q}}}
\newcommand{\xn}{x_n^{\phantom{q}}}
\newcommand{\vex}{\vec{x}}
\newcommand{\vexo}{\vex_1^{\phantom{q}}}
\newcommand{\vexn}{\vex_n^{\phantom{q}}}
\newcommand{\vexk}{\vex_k^{\phantom{q}}}

\newcommand{\vey}{\vec{y}}

\newcommand{\veyk}{\vey_k^{\phantom{q}}}

\newcommand{\pk}{p_k^{\phantom{q}}}
\newcommand{\pkq}{p_{k}^q}
\newcommand{\pj}{p_j^{\phantom{q}}}

\newcommand{\po}{p_1^{\phantom{q}}}
\newcommand{\pt}{p_2^{\phantom{q}}}

\newcommand{\pd}{p_{\dime}^{\phantom{q}}}
\newcommand{\lpj}{{}^{l\hspace{-0.3mm}}p_j^{\phantom{q}}}

\newcommand{\tpj}{\tilde{p}_j^{\phantom{q}}}
\newcommand{\tpk}{\tilde{p}_k^{\phantom{q}}}
\newcommand{\tpo}{\tilde{p}_1^{\phantom{q}}}

\newcommand{\tpd}{\tilde{p}_{\dime}^{\phantom{q}}}

\newcommand{\rk}{r_{\hspace{-0.2mm}k}^{\phantom{q}}}
\newcommand{\ro}{r_{\hspace{-0.2mm}1}^{\phantom{q}}}
\newcommand{\rt}{r_{\hspace{-0.2mm}2}^{\phantom{q}}}
\newcommand{\rn}{r_{\hspace{-0.2mm}n}^{\phantom{q}}}
\newcommand{\rkq}{r_{\hspace{-0.2mm}k}^q}
\newcommand{\sk}{s_{k}^{\phantom{q}}}
\newcommand{\so}{s_{1}^{\phantom{q}}}

\newcommand{\sn}{s_{n}^{\phantom{q}}}

\newcommand{\vep}{\vec{p}\hspace{0.5mm}}
\newcommand{\veprho}{\vep(\hrho)}
\newcommand{\vepkrho}{\vec{p}_k^{\phantom{q}}(\hrho)}
\newcommand{\vepho}{\vep(\ho)}

\newcommand{\bm}{\mathsf{M}}
\newcommand{\cml}{\mathsf{m}_l}
\newcommand{\perm}{\mathsf{P}}
\newcommand{\perml}{{}^{l\hspace{-0.2mm}}\mathsf{P}}
\newcommand{\permljk}{{}^{l\hspace{-0.2mm}}\mathsf{P}_{\hspace{-0.6mm}j,k}}

\newcommand{\qdm}{\mathfrak{Q}}
\newcommand{\qds}{\qdm_t}
\newcommand{\qdss}{\qdm_s}
\newcommand{\qdsts}{\qdm_{t+s}}

\newcommand{\opf}{\hat{F}}
\newcommand{\opfk}{\opf_k^{\phantom{\ast}}}
\newcommand{\opfka}{\opf_k^\ast}
\newcommand{\opfo}{\opf_{1}^{\phantom{\ast}}}
\newcommand{\opfd}{\opf_{\dime^2-1}^{\phantom{\ast}}}

\newcommand{\opk}{\hat{K}}
\newcommand{\opkk}{\opk_k^{\phantom{\ast}}}
\newcommand{\opkj}{\opk_j^{\phantom{\ast}}}
\newcommand{\opkka}{\opk_k^\ast}
\newcommand{\opkja}{\opk_j^\ast}
\newcommand{\opko}{\opk_{1}^{\phantom{\ast}}}
\newcommand{\opkm}{\opk_{m}^{\phantom{\ast}}}

\newcommand{\opl}{\hat{L}}
\newcommand{\oplk}{\opl_k^{\phantom{\ast}}}
\newcommand{\oplks}{\opl_k^{2}}
\newcommand{\oplo}{\opl_{1}^{\phantom{\ast}}}
\newcommand{\opld}{\opl_{\dime^2-1}^{\phantom{\ast}}}

\newcommand{\den}{\mathscr{N}}

\newcommand{\uj}{\hat{U}_{\hspace{-0.3mm}j}^{\phantom{\ast}}}
\newcommand{\ujast}{\hat{U}_{\hspace{-0.3mm}j}^{\ast}}
\newcommand{\uo}{\hat{U}_{1}^{\phantom{\ast}}}
\newcommand{\un}{\hat{U}_{\hspace{-1.2mm}\den}^{\phantom{\ast}}}
\newcommand{\um}{\hat{U}_{\hspace{-0.3mm}m}(t)}
\newcommand{\umast}{\hat{U}_{\hspace{-0.3mm}m}(t)^{\ast}}
\newcommand{\wem}{w_m(t)}

\newcommand{\fin}{\hspace{0.5mm}}

\newcommand{\spa}{\hspace{-2mm}}
\newcommand{\tre}{\hspace{0.3mm}}
\newcommand{\quattro}{\hspace{0.4mm}}
\newcommand{\cinque}{\hspace{0.5mm}}
\newcommand{\sei}{\hspace{0.6mm}}
\newcommand{\sette}{\hspace{0.7mm}}
\newcommand{\otto}{\hspace{0.8mm}}

\newcommand{\mtre}{\hspace{-0.3mm}}

\newcommand{\mcinque}{\hspace{-0.5mm}}
\newcommand{\msei}{\hspace{-0.6mm}}

\newcommand{\motto}{\hspace{-0.8mm}}
\newcommand{\mdieci}{\hspace{-1mm}}
\newcommand{\mdodici}{\hspace{-1.2mm}}

\newcommand{\erre}{\mathbb{R}}

\newcommand{\ccc}{\mathbb{C}}
\newcommand{\de}{\mathrm{d}}
\newcommand{\dert}{\frac{\de\phantom{t}}{\de t}}
\newcommand{\eee}{\mathrm{e}}
\newcommand{\bsp}{\mathfrak{X}}
\newcommand{\csg}{\mathfrak{C}}
\newcommand{\opa}{\hat{A}}
\newcommand{\opb}{\hat{B}}

\newcommand{\dom}{\mathrm{Dom}}

\newcommand{\oph}{\hat{H}}
\newcommand{\gene}{\mathfrak{L}}
\newcommand{\genea}{\mathfrak{L}^{\ast}}

\newcommand{\bH}{\mathcal{B}(\mathcal{H})}

\newcommand{\convo}{\hspace{-0.5mm}\circledcirc\hspace{-0.3mm}}
\newcommand{\mut}{\mu_t}

\newcommand{\nut}{\nu_t}
\newcommand{\su}{\varsigma}
\newcommand{\sua}{\hspace{0.3mm}|\su|}
\newcommand{\supo}{\varsigma_{\mbox{\tiny $+$}\hspace{-0.2mm}}}
\newcommand{\sune}{\varsigma_{\mbox{\tiny $-$}\hspace{-0.2mm}}}
\newcommand{\sut}{\varsigma_t}
\newcommand{\sus}{\varsigma_s}
\newcommand{\sust}{\varsigma_{s+t}}
\newcommand{\suta}{\hspace{0.3mm}|\sut|}
\newcommand{\suz}{\varsigma_0}

\newcommand{\ima}{\mathrm{i}}

\newcommand{\cpm}{\mathfrak{F}}
\newcommand{\cpmo}{\mathfrak{F}_1^{\phantom{\ast}}}
\newcommand{\cpmt}{\mathfrak{F}_2^{\phantom{\ast}}}
\newcommand{\cpmoa}{\mathfrak{F}_1^{\ast}}
\newcommand{\cpmta}{\mathfrak{F}_2^{\ast}}

\newcommand{\sopa}{\mathfrak{A}}

\newcommand{\trc}{\mathcal{B}_1(\mathcal{H})}

\newcommand{\slim}{\mbox{s\hspace{0.3mm}-}\hspace{-0.7mm}\lim}

\newcommand{\randu}{\mathfrak{U}\hspace{0.3mm}}
\newcommand{\randut}{\mathfrak{V}\hspace{0.3mm}}

\begin{document}

\title{Characterizing the dynamical semigroups that do not decrease a quantum entropy}

\author{
Paolo Aniello$^1$ and Dariusz Chru\'sci\'nski$^2$
\vspace{2mm} \\
\small \it $^1$Dipartimento di Fisica ``Ettore Pancini'', Universit\`a di Napoli ``Federico II'',
\\ \small \it and Istituto Nazionale di Fisica Nucleare (INFN), Sezione di Napoli,
\\ \small \it Complesso Universitario di Monte S.\ Angelo, via Cintia, I-80126 Napoli, Italy
\vspace{1mm} \\ \small \it $^2$Institute of Physics, Faculty of Physics, Astronomy and Informatics
\\ \small \it Nicolaus Copernicus University
\\ \small \it Grudziadzka 5, 87–100 Toru\'n, Poland
}

\date{}

\maketitle

\begin{abstract}
\noindent  In finite dimensions, we provide characterizations of the quantum dynamical semigroups
that do not decrease the von~Neumann, the Tsallis and the R\'enyi entropies, as well as a family
of functions of density operators strictly related to the Schatten norms. A few remarkable consequences
--- in particular, a description of the associated infinitesimal generators --- are derived, and
some significant examples are discussed. Extensions of these results to semigroups of trace-preserving positive
(i.e., not necessarily completely positive) maps and to a more general class of quantum entropies are also considered.
\end{abstract}

\section{Introduction}
\label{intro}

Entropy is one of the most fundamental and ubiquitous concepts in science.
In particular, it plays a central role in the theory of open quantum systems~{\cite{Breuer}}
and in quantum information theory~{\cite{Nielsen}}, where various entropy functions can be considered;
e.g., the von~Neumann, the Tsallis and the R\'enyi entropies (see~{\cite{Bengtsson}} and references therein).
An interesting problem is to characterize
the temporal evolution of a certain quantum entropy; e.g., to ascertain whether a certain class
of dynamics of open quantum systems does not decrease this quantity (for all initial states).

As is well known, under certain assumptions the evolution of an open quantum system
can be described by a suitable class of semigroups of operators; i.e., the class of
\emph{quantum dynamical semigroups}~{\cite{Breuer,Holevo,Alicki}}. A natural problem is then to
characterize the subclass of quantum dynamical semigroups (say, for a finite-dimensional quantum system)
that do not decrease a quantum entropy. A problem of this kind has been investigated,
for the von~Neumann entropy and using a computational approach, by Banks, Susskind and Peskin~{\cite{Banks}}
(on the base of a proposal of Hawking), inspired by the idea that
in order to accommodate gravity in the context of quantum theory one may allow
pure states --- actually, of a \emph{closed} system --- to evolve into mixed states,
so envisioning possible modifications of standard quantum mechanics.

When approaching this problem, one may adopt two equivalent points of view:
to study the properties of quantum dynamical semigroups themselves (`integrated approach')
or of their infinitesimal generators (`master equation approach'). The generators
of quantum dynamical semigroups admit a complete classification, the
so called Gorini-Kossakowski-Lindblad-Sudarshan canonical form~{\cite{Gorini,Lindblad}}.
Therefore, in the master equation approach solving the problem amounts to characterizing,
within this general classification, the typical form of those generators
that give rise to a temporal evolution enjoying the aforementioned property.

In this paper, we will switch back and forth between the two points of view.
We will adopt the integrated approach first, obtaining various characterizations
of those (finite-dimensional) quantum dynamical semigroups that
do not decrease the von~Neumann, the Tsallis and the R\'enyi entropies, as well as
a family of functions of density operators which are directly related to the Schatten norms.
The simplest characterization is given by the --- both necessary and sufficient --- condition of
unitality of the maps forming the dynamical semigroups.
As a byproduct, a characterization of the associated infinitesimal generators
will be then easily derived.

It is worth observing that the class of quantum dynamical semigroups
that do not decrease the mentioned families of quantum entropies
turns out to contain, as a remarkable subclass, the so-called (finite-dimensional)
\emph{twirling semigroups}~{\cite{Aniello1,Aniello2,Aniello3,Aniello-OvF,AnielloCQS,AnielloPT,AnielloFPTB}}.
These semigroups of operators are characterized by a suitable integral
expression involving a representation of a locally compact group and
a convolution semigroup of probability measures on that group.
We will argue that actually, in the case where the Hilbert space of the
quantum system is two-dimensional, the twirling semigroups are
\emph{precisely} the quantum dynamical semigroups that do not decrease
these quantum entropies (since the former class of quantum dynamical
semigroups --- and, in dimension two, the latter as well ---
can be shown to coincide with the class of random unitary semigroups).

These results admit various --- more or less straightforward --- generalizations.
First, one can relax the requirement that the semigroups of operators considered
be completely positive; precisely, one can deal, more generally, with semigroups of trace-preserving,
positive linear maps. In this regard, we stress that, although complete positivity
is often regarded as a fundamental property, its justification on the physical ground
is controversial~{\cite{Shaji}}. An interesting fact regarding the dynamical semigroups, in this
more general class, that do not decrease the mentioned quantum entropies is
that they admit an integral expression which can be regarded as a generalization
of the formula characterizing the twirling semigroups (in the general case,
semigroups of probability measures are replaced with certain families
of signed measures). Moreover, the associated infinitesimal generators form a suitable convex cone,
which, in the case where the Hilbert space of the quantum system is two-dimensional,
will be described in detail.

Another interesting generalization stems from the observation that
that a certain kind of quantum entropies verifying two fundamental axioms
(see sect.~{\ref{conclusions}}) are not decreased by the same class of
semigroups of trace-preserving positive maps considered above.

The paper is organized as follows. In sect.~{\ref{basic}}, we recall some basic facts,
fix the notations and define the classes of quantum entropies that are considered
in the paper. Next, in sect.~{\ref{main}}, we prove the main results concerning the
completely positive dynamical semigroups, and we provide some significant examples.
Extensions of these results to semigroups of trace-preserving, positive --- but not necessarily completely
positive --- maps are investigated in sect.~{\ref{further}}, and the case of a qubit system
is studied in detail. Finally, in sect.~{\ref{conclusions}}, a few conclusions are drawn,
and a general kind of quantum entropies that are not decreased precisely by the class of
dynamical semigroups considered in the paper is defined.

\section{Basic facts and notations}
\label{basic}

In this section, we will establish the main notations, and we will briefly recall some
basic facts that will be useful in the rest of the paper, including the
definition of the quantum entropies we are interested in and some of their salient
properties.

Let $\bsp$ be a separable real or complex Banach space. Then, denoting by
$\errep$ the set of non-negative real numbers (we also set
$\erreps\equiv\errep\smallsetminus\{0\}$), a family
$\{\csg_t\}_{t\in\errep}$ of bounded linear operators in $\bsp$ is
said to be a (continuous) \emph{semigroup of operators} if the following conditions
are satisfied~\cite{Hille,Yosida-book}:
\begin{enumerate}
\item $\csg_t\cinque \csg_s = \csg_{t+s}$,
$\forall\quattro t,s\ge0$;

\item $\csg_0 = \ide$;

\item $\lim_{t\downarrow 0}\|\csg_t\cinque\zeta -\zeta\|=0$,
$\forall\quattro\zeta\in\bsp$; i.e., $\slim_{t\downarrow
0}\csg_t=\ide$ \  (strong right continuity at $t=0$).
\end{enumerate}
The semigroup of operators $\{\csg_t\}_{t\in\errep}$ is completely characterized by its ---
in general, densely defined --- \emph{infinitesimal generator}; namely, by the closed linear
operator $\sopa$ in $\bsp$ defined by
\begin{equation}
\dom(\sopa)\defi\Big\{\zeta\in\bsp\colon \exists\cinque
\lim_{t\downarrow 0}
t^{-1}\big(\csg_t\cinque\zeta-\zeta\big)\Big\} \fin ,\ \ \
\sopa\cinque\zeta \defi \lim_{t\downarrow 0}
t^{-1}\big(\csg_t\cinque\zeta-\zeta\big) \fin ,\
\forall\quattro\zeta\in\dom(\sopa) \fin .
\end{equation}

In this paper, we will consider the special class of semigroups of
operators consisting of \emph{quantum dynamical maps} (or quantum channels) in $\trc$,
the Banach space of trace class operators in a complex separable Hilbert space $\hh$;
namely, the class of \emph{quantum dynamical semigroups}~{\cite{Breuer,Holevo,Alicki}}
--- or `completely positive dynamical semigroups' --- acting on $\trc$.
Clearly, the fundamental assumption that the convex body $\stah\subset\trc$
of density operators in $\hh$ --- the unit trace, positive trace class operators ---
be mapped into itself entails, more generally, positivity (rather than complete positivity);
thus, the larger class of semigroups of trace-preserving, positive linear maps will be considered as well.
In the following, we will actually deal with a finite-dimensional
--- say, with a $\dime$-dimensional, $\dime\ge 2$ --- complex Hilbert space $\hh$. It is clear
that, in this case, $\trc=\bH$ (the space of linear operators in $\hh$,
or $\dime\times\dime$ complex matrices). The identity in $\bH$ will be denoted by $\id$.
The unitary group $\unihh$ of $\hh$ will be regarded as endowed with
the strong topology (which, however, in the finite-dimensional case is equivalent to
the topology induced by any norm topology in $\bH$).

A quantum dynamical semigroup $\{\qds\colon\bH\rightarrow\bH\}_{t\in\errep}$
is then characterized by its --- in this case, of course, bounded --- generator $\gene$:
\begin{equation} \label{QDSG}
\gene = \lim_{t\downarrow 0}
t^{-1}\big(\qds - \ide) \fin , \ \ \ \qds = \exp(\gene\sei t) \fin .
\end{equation}
According to the celebrated Gorini-Kossakowski-Lindblad-Sudarshan classification
theorem~{\cite{Gorini,Lindblad}} (see also~\cite{Alicki} and the recent review~\cite{DC}),
$\gene$ has --- in the Schr\"odinger representation ---
the general form
\begin{equation} \label{forgene}
\gene\sei\opa = - \ima \big[\oph, \opa\big] +
\cpm\sei\opa - \frac{1}{2}
\left(\big(\cpm^\ast \id\big)\tre\opa+\opa\sette\big(\cpm^\ast \id\big)\right),
\end{equation}
where $\oph$ is a traceless selfadjoint operator in $\hh$,
$\cpm \colon \bH\rightarrow\bH$ a completely positive linear map~{\cite{Bengtsson,Holevo,Paulsen}} and
$\cpm^\ast$ its adjoint with respect to the Hilbert-Schmidt scalar product $\hspro$ in $\bH$.

\begin{remark} \label{pairi}
A linear map $\lino$ acting in $\bH$ is said to be
\emph{adjoint-preserving} if $\lino\sei\opa^\ast=\big(\lino\sei\opa\big)^\ast$,
for every $\opa\in\bH$. A positive (in particular, a completely positive) linear map and
the generator of a semigroup of positive  maps in $\bH$ are adjoint-preserving. Hence, for
defining the adjoint of such maps it is not relevant whether the pairing
$\big(\opa,\opb\big)\mapsto\tr\big(\opa\sei\opb\big)$ or
$\big(\opa,\opb\big)\mapsto\tr\big(\opa^\ast\opb\big)\ifed\big\langle\opa,\opb\big\ranglehs$
is used.
\end{remark}

\begin{remark} \label{contra}
One can easily show that, for a quantum dynamical map $\qdm$, denoting by $\|\cdot\noro$
the trace norm  --- in general, in the following $\|\cdot\norp$ will denote
the Schatten $\pap$-norm in $\bH$,
\begin{equation}
\big\|\opa\big\norp\defi\tr\big(\big|\opa\big|^\pap\big)^{1/\pap} ,
\ \ \ 1\le\pap<\infty \fin ,
\end{equation}
(thus, $\|\cdot\|_2\equiv\|\cdot\norhs$), and $\|\cdot\nori$ the standard operator norm ---
$\big\|\qdm\sei\opa\big\noro\le\big\|\opa\big\noro$ (note, however, that in the proof
positivity, rather than complete positivity, together with preservation of the trace,
is the relevant property). Therefore, regarding $\bH=\trc$ as endowed with the trace norm,
a quantum dynamical semigroup --- more generally, a semigroup of trace-preserving positive maps
--- is a \emph{contraction semigroup}.
\end{remark}

For every density operator $\hrho\in\stah$,
we can define the \emph{von~Neumann entropy} of $\hrho$, namely,
\begin{equation} \label{defvne}
\vne(\hrho) \defi -\tr(\hrho \ln \hrho) = - \sum_{k=1}^{\dime} \pk \ln \pk \fin ,
\end{equation}
where $\{\po,\ldots,\pd\}$ is the whole set of the eigenvalues of $\hrho$,
repeated according to degeneracy (of course, here $0\ln 0\equiv 0$).

\begin{remark}
Clearly, the definition of $\vne(\hrho)$ --- and of all other entropies
that will be considered below --- does not depend on the way the eigenvalues
$\{\pk\}$ of $\hrho$ are ordered. On the other hand,
arranging the sum on the rhs of~{(\ref{defvne})}
in such a way that the eigenvalues are labeled in decreasing
order --- $\po\ge\cdots\ge\pd\ge 0$ --- we obtain the \emph{ordered eigenvalue vector}
$\vep=\veprho=(\po,\ldots,\pd)$ of $\hrho$. This convention can be regarded as a
way for removing the ambiguity when labeling the eigenvalues of a density operator.
\end{remark}

\begin{remark}
Recall that ordered vectors are also used when defining the \emph{majorization}
relation~{\cite{Bhatia,Marshall}} between vectors. Given a vector
$x=(\xo,\ldots,\xn)$ in $\erren$, we denote by $\vex=(\vexo,\ldots,\vexn)$
the vector obtained rearranging the coordinates of $x$ in decreasing order:
$\vexo\ge\cdots\ge\vexn$. For $x,y\in\erren$, one says that $x$ is majorized by $y$
--- in symbols, $x\prec y$ --- if
\begin{equation} \label{majo}
\sum_{k=1}^j \vexk \le \sum_{k=1}^j \veyk \fin , \ \ 1\le j\le n \fin ,
\ \ \mbox{and} \ \ \sum_{k=1}^n \vexk = \sum_{k=1}^n \veyk \fin .
\end{equation}
Clearly, if $\xo\ge 0,\ldots,\xn\ge 0$ and $\sum_{k=1}^n \xk = 1$
(namely, if $\{\xk\}$ is a probability distribution), then
\begin{equation} \label{mmcase}
(1/n,\ldots,1/n)\prec (\xo,\ldots,\xn) \prec (1,0,\ldots,0) \fin .
\end{equation}
Moreover, $x\prec y$ and $y\prec x$ implies that $x \approx y$ ---
i.e., $x=\perm \sei y$, for some permutation matrix $\perm$ --- and $\vex=\vey$.
Majorization between eigenvalue vectors induces a natural
majorization relation between density operators~{\cite{Bengtsson,Wehrl}}, i.e.,
\begin{equation}
\ho \prec \hrho \ \defar \ \vepho\prec\veprho \fin .
\end{equation}
By relation~{(\ref{mmcase})}, the \emph{maximally mixed state} in $\stah$ --- namely,
the state $\mms\defi\dime^{-1}\tre\id$ --- is majorized by any other state, whereas
every state is majorized by any pure state. In general, for $\ho \prec \hrho$
and $\hrho\nprec\ho$, the state
$\ho$ will be `less pure' or `more mixed' than $\hrho$~{\cite{Wehrl}}.
\end{remark}

We will also consider the \emph{Tsallis entropy} $\tseq$, parametrized by $q$,
with $0<q\neq 1$; i.e.,
\begin{equation}
\tseq(\hrho) \defi \frac{1}{1-q} \otto \tr(\hrhoq) = \teq(\{\pk\}) \fin,
\end{equation}
where, for every finite probability distribution $\{\rk\}$,
\begin{equation}
\teq(\{\rk\}) \defi \frac{1}{1-q} \left(\sum_k \rkq -1\right).
\end{equation}
Clearly, as previously noted, $\teq(\{\rk\})$ does not depend on the ordering
of the elements of the probability distribution $\{\rk\}$; i.e., $\teq$ can also be
regarded as a symmetric function of the probability vector $(\ro,\rt,\ldots)$.
We now set
\begin{equation}
\teo(\{\rk\})\defi \lim_{q\rightarrow 1}\teq(\{\rk\})\fin ,
\end{equation}
and we observe that
\begin{equation}
\vne(\hrho) = \lim_{q\rightarrow 1} \tseq(\hrho)=\teo(\{\pk\}) \fin .
\end{equation}
Considering this relation, it is also natural to set
\begin{equation}
\tseo(\hrho) \equiv \vne(\hrho) \fin .
\end{equation}
Note that, for $q=2$, the Tsallis entropy is strictly related to
the \emph{purity} since $\tset(\hrho)=1-\tr(\hrhot)$; $\tset$ is
sometimes called the \emph{linear entropy}.
In the following, we will therefore consider the full family of entropies $\{\tseq\}_{q> 0}$.
We will use the well known fact that the associated functions $\{\teq\}_{q> 0}$
are \emph{concave} (see, e.g.,~{\cite{Naudts}}); namely,
for every pair $(\ro,\ldots,\rn)$, $(\so,\ldots,\sn)$ of probability vectors, and every convex combination
$(\epsilon\tre\ro+(1-\epsilon)\so,\ldots,\epsilon\tre\rn+(1-\epsilon)\sn)$, $\epsilon\in[0,1]$,
we have:
\begin{equation} \label{conca}
\teq(\{(\epsilon\tre\rk+(1-\epsilon)\sk)\}) \ge
\epsilon\tre\teq(\{\rk\})+(1-\epsilon)\tre\teq(\{\sk)\}) \fin .
\end{equation}

\begin{remark}
For the von~Neumann entropy, a stronger property holds.
Using Klein's inequality~{\cite{Bengtsson}},
one can show that, for every pair $\hrho,\ho\in\stah$,
$\vne(\epsilon\tre\hrho+(1-\epsilon)\tre\ho)\ge\epsilon\tre\vne(\hrho)+(1-\epsilon)\tre\vne(\ho)$.
Otherwise stated, $\vne$ itself, as a function on the convex body $\stah$, is concave.
\end{remark}

Another interesting family of entropies are the \emph{R\'enyi entropies} $\{\req\}_{q> 0}$.
For $0<q\neq 1$, we set
\begin{equation}
\req(\hrho) \defi \frac{1}{1-q} \otto \ln\tr(\hrhoq) =
\frac{1}{1-q} \ln\mcinque\left(\sum_{k=1}^{\dime} \pkq \right)\ifed\roq(\{\pk\}) \fin ,
\end{equation}
where $\{\pk\}$ is, as above, the whole set of the eigenvalues of $\hrho$, repeated according to
degeneracy. Once again we have:
\begin{equation}
\reqo (\hrho)  \equiv \vne(\hrho) = \lim_{q\rightarrow 1} \req(\hrho)=
\lim_{q\rightarrow 1} \roq(\{\pk\}) \ifed \roo(\{\pk\}) \fin .
\end{equation}
Moreover, for $q=2$, also the R\'enyi entropy is directly related to
the purity: $\reqt(\hrho)=-\ln\tr(\hrhot)$. The entropies $\{\req\}_{q> 0}$
are known --- see, e.g.,~{\cite{Bullen}} --- to be \emph{Schur concave}, namely,
for every pair $r=(\ro,\ldots,\rn)$, $s=(\so,\ldots,\sn)$ of probability vectors,
such that $s\prec r$ ($s$ is majorized by $r$),
$\roq(\{\sk\})\ge\roq(\{\rk\})$. Therefore, we have that
\begin{equation} \label{schuconca}
\ho \prec \hrho \ \defar \ \vepho\prec\veprho
\ \Longrightarrow \req (\ho) \ge \req (\hrho) \fin .
\end{equation}

\begin{remark} \label{symschu}
As in the case of $\teq$, also $\roq$ can be regarded as a symmetric function on
a probability simplex. This fact is actually related to its Schur concavity,
because every Schur concave function on a symmetric domain is symmetric
(w.r.t.\ permutations of its arguments)~{\cite{Marshall}}.
\end{remark}

\begin{remark}
A symmetric function $\varphi\colon\simpl\rightarrow\erre$, on a probability simplex
$\simpl$, is \emph{strictly} Schur concave if $x \prec y$, with $x \noappro y$ (i.e., $\vex\neq\vey$),
implies that $\varphi(x)>\varphi(y)$. Accordingly, we will say that a function $\fun\colon\stah\rightarrow\erre$
is strictly Schur concave if $\ho\prec\hrho$, with $\ho\noappro\hrho$ (i.e., $\vepho\neq\veprho$), implies that
$\fun(\hrho)>\fun(\ho)$. It turns out --- see sect.~{\ref{conclusions}} --- that
the R\'enyi entropies $\{\req\}_{q> 0}$ are actually strictly Schur concave.
\end{remark}

Finally, consider the function
\begin{equation}
\nep\colon\stah\ni\hrho\mapsto(1-\|\hrho\norp)\in\erre \fin ,
\ \ \ 1<\pap\le\infty \fin .
\end{equation}
Observe that, using the same kind of notation as above, we have:
\begin{equation}
\nep(\hrho)= 1-\left(\sum_k \pkp \right)^{\mdieci 1/\pap}
\mdodici\ifed\alp(\{\pk\})\fin , \ \ 1<\pap<\infty \fin ,
\end{equation}
\begin{equation} \label{ainf}
\neinf(\hrho)=1-\max(\{\pk\})\ifed\alinf(\{\pk\})
= \lim_{\pap\rightarrow\infty} \alp(\{\pk\})  \fin .
\end{equation}
The following inequalities hold:
\begin{equation}
n^{\frac{1-\pap}{\pap}}\le \left(\sum_k \pkp \right)^{\mdieci 1/\pap}
\mdodici \le 1 \fin , \ \ 1<\pap<\infty \fin , \ \ \ 1/n\le\max(\{\pk\})\le 1 \fin .
\end{equation}
It follows that
\begin{equation}
0\le\alp(\{\pk\})\le 1-n^{\frac{1-\pap}{\pap}} \fin , \ \ 1<\pap<\infty \fin ,
\ \ \ 0\le \alinf(\{\pk\}) \le 1-1/n \fin ,
\end{equation}
where in both cases the upper inequality is saturated by the probability distribution
$(1/n,\ldots,1/n)$ only, whereas the lower one is saturated by any
probability distribution $\{\pk\}$ such that the corresponding ordered probability vector
is of the form $\vep=(1,0,\ldots,0)$. Moreover, by the triangle inequality of
the norm $\|\cdot\norp$, $1<\pap\le\infty$, one easily finds
\begin{equation}
\nep(\epsilon\tre\hrho+(1-\epsilon)\tre\ho)
\ge\epsilon\tre\nep(\hrho)+(1-\epsilon)\tre\nep(\ho) \fin ,
\ \ \ \hrho,\ho\in\stah \fin , \ \epsilon\in[0,1] \fin ,
\end{equation}
and this relation implies for $\alp$ a concavity property analogous
to~{(\ref{conca})}. Therefore, the family of functions $\{\nep\}_{1<\pap\le\infty}$
enjoys some of the main typical properties of a quantum entropy.\footnote{Here, we will
not deal with the much debated question of what fundamental properties a quantum entropy
should satisfy.}

In addition to the previously recalled properties, in the following it will be of
central importance the fact that the maximally mixed state $\mms$ maximizes
--- strictly and globally --- all the entropies $\{\tseq\}_{q> 0}$, $\{\req\}_{q> 0}$
and $\{\nep\}_{1<\pap\le\infty}$.

\section{The completely positive dynamical semigroups that do not decrease a quantum entropy}
\label{main}

We start this section by assigning a precise meaning to the statement that
a quantum channel or a quantum dynamical semigroup does not decrease a quantum
entropy.

For the sake of conciseness, in this section and in the subsequent one
\emph{we will denote by $\entro$ any of the quantum entropies
$\{\tseq\}_{q> 0}$, $\{\req\}_{q> 0}$, $\{\nep\}_{1<\pap\le\infty}$.}

\begin{definition} \label{donode}
We say that a quantum dynamical map (a completely positive, trace-preserving linear map)
$\qdm\colon\bH\rightarrow\bH$ does not decrease the entropy $\entro$ if
\begin{equation}
\entro(\qdm\tre\hrho) \ge \entro(\hrho) \fin ,
\end{equation}
for all $\hrho\in\stah$.
\end{definition}

\begin{definition} \label{sedono}
We say that a quantum dynamical semigroup $\{\qds\colon\bH\rightarrow\bH\}_{t\in\errep}$ does not decrease
the entropy $\entro$ if
\begin{equation} \label{deca}
\entro(\qdsts\tre\hrho) \ge \entro(\qds\tre\hrho) \fin ,
\end{equation}
for all $\hrho\in\stah$ and all $t,s\ge 0$.
\end{definition}
A slightly simpler version of the latter definition is obtained taking into account
the following elementary fact.
\begin{proposition} \label{cohe}
A quantum dynamical semigroup $\{\qds\colon\bH\rightarrow\bH\}_{t\in\errep}$ does not decrease
the entropy $\entro$ if and only if
\begin{equation} \label{decb}
\entro(\qds\tre\hrho) \ge \entro(\hrho) \fin ,
\end{equation}
for all $\hrho\in\stah$ and all $t\ge 0$. Otherwise stated, a quantum dynamical semigroup
does not decrease the entropy $\entro$ if and only if each of its members,
as a quantum dynamical map, enjoys this property.
\end{proposition}

\begin{proof}
If relation~{(\ref{deca})} holds, then setting $t=0$ therein one gets~{(\ref{decb})}.
Conversely, if~{(\ref{decb})} holds, then
\begin{equation}
\entro(\qdsts\tre\hrho) = \entro(\qdss(\qds\tre\hrho))\ge \entro(\qds\tre\hrho) \fin ,
\end{equation}
for all $\hrho\in\stah$ and all $t,s\ge 0$.
\end{proof}

The two following results will be central in the proof of the main theorem.
Recall that a real square matrix is called \emph{bistochastic} (or doubly stochastic)
if it has non-negative entries, and each row and each column sums up to
unity.

\begin{lemma} \label{teclem}
Let $\qdm\colon\bH\rightarrow\bH$ be a quantum dynamical map. Then, the
following properties are equivalent:
\begin{description}

\item[\tt (P1)]
$\qdm$ is unital, namely,
\begin{equation}
\qdm\sei \id = \id  ;
\end{equation}

\item[\tt (P2)]
for every $\hrho\in\stah$, there is a bistochastic matrix $\bm$ such that
$\vep(\qdm\sei\hrho)=\bm\sei\vep(\hrho)$;

\item[\tt (P3)]
for every $\hrho\in\stah$, $\qdm\sei\hrho\prec\hrho$.

\end{description}
\end{lemma}

\begin{proof}
{\tt (P1)} implies {\tt (P2)}. Indeed, consider the
eigenvalue decompositions
\begin{equation} \label{eide}
\hrho=\sum_{k=1}^{\dime} \pk \sei \prok \fin , \ \ \
\qdm\sei\hrho=\sum_{j=1}^{\dime} \tpj \sei \vproj \fin ,
\end{equation}
where $\{\prok\}$, $\{\vproj\}$ are rank-one projections.
Then, the probability distribution $\{\tpj\}$ is given by
\begin{equation}
\tpj=\tr\big(\vproj(\qdm\sei\hrho)\big) =
\sum_{k=1}^{\dime} \tr\big(\vproj(\qdm\sei\prok)\big) \tre \pk
\ifed \sum_{k=1}^{\dime}\bm_{jk}\sei\pk \fin .
\end{equation}
By the arbitrariness of the probability distribution $\{\pk\}$
(i.e., by varying $\hrho$) and the fact that
$\qdm\sei\mms=\mms$ ($\qdm$ is unital), we see that the matrix
$\bm$ (with positive entries) must be bistochastic.
By assuming now that the probability distributions
$\{\pk\}$, $\{\tpk\}$ in~{(\ref{eide})}
are arranged in decreasing order,
we get $\vep(\qdm\sei\hrho)=\bm\sei\vep(\hrho)$.

Moreover, {\tt (P2)} implies {\tt (P3)}. In fact, by a well known result
on bistochastic matrices (see~{\cite{Bhatia}}, Theorem~{II.1.9}), we have that
\begin{equation}
\bm \ \mbox{bistochastic} \ \Longrightarrow \ \vep(\qdm\sei\hrho)=\bm\sei\vep(\hrho)\prec\veprho
\ \ (\Longleftrightarrow \qdm\sei\hrho \prec \hrho) \fin .
\end{equation}

Finally, {\tt (P3)} implies {\tt (P1)} because $\qdm\sei\mms\prec\mms$ entails
that $\qdm\sei\mms=\mms$. The proof is complete.
\end{proof}

\begin{lemma} \label{mainlem}
A quantum dynamical map $\qdm\colon\bH\rightarrow\bH$ does not decrease
the entropy $\entro$ if and only if it is unital.
\end{lemma}

\begin{proof}
Assume that quantum dynamical map $\qdm$ is unital, and take any $\hrho\in\stah$.
Let $\veprho=(\po,\ldots,\pd)$ be the ordered eigenvalue vector of $\hrho$,
and let $\vepho=(\tpo,\ldots,\tpd)$ be the analogous vector for $\ho\equiv \qdm\tre\hrho$.
Then, since the quantum dynamical map $\qdm$ is unital,
by Lemma~{\ref{teclem}} there is a bistochastic
$\dime\times\dime$ matrix $\bm$ such that $\tpj=\sum_{k}\bm_{jk}\sei\pk$.

Suppose first that $\entro=\tseq$, for some $q > 0$.
By Birkhoff-von~Neumann theorem (see~{\cite{Bhatia}}, Theorem~{II.2.3}),
the matrix $\bm$ can be expressed as a convex combination
of permutation matrices
\begin{equation}
\bm = \sum_l \cml \sei \perml \fin , \ \ \cml >0 \fin , \ \ \sum_l \cml =1 \fin ;
\end{equation}
hence:
\begin{equation}
\tpj = \sum_l \cml \sei \lpj \fin , \ \ \ \lpj\defi\sum_k \permljk\sei\pk \fin .
\end{equation}
Then, by the concavity property~{(\ref{conca})} of $\teq$ and by the fact that $\{\pj\}=\big\{\lpj\big\}$
(regarded as un-ordered sets), we find
\begin{equation}
\tseq(\ho) = \teq(\{\tpj\}) \ge \sum_l \cml \sei \teq\big(\big\{\lpj\big\}\big)
= \sum_l \cml \sei \teq(\{\pj\}) = \teq(\{\pj\})=\tseq(\hrho) \fin .
\end{equation}

The same proof can be repeated verbatim for $\entro=\nep$ (replacing $\teq$
with $\alp$), $1<\pap\le\infty$.

Suppose now that $\entro=\req$, for some $q> 0$. Then,
by Lemma~{\ref{teclem}} we have:
\begin{equation}
\qdm \ \mbox{unital} \ \Longrightarrow \ \ho \prec \hrho \fin .
\end{equation}
Therefore, by the Schur concavity of the R\'enyi entropies, we conclude
that
\begin{equation}
\qdm \ \mbox{unital} \ \Longrightarrow
\ho \prec \hrho \ \Longrightarrow \req (\ho) \ge \req (\hrho) \fin .
\end{equation}

Conversely, suppose that $\qdm$ does not decrease $\entro$. Since the strict global maximum for this
quantity is attained at $\mms=\dime^{-1}\tre \id$ (the maximally mixed state), it follows
that $\qdm\tre \id=\id$.
\end{proof}

\begin{remark} \label{onlypos}
Observe that $\qdm$ can be replaced --- in Definition~{\ref{donode}}, and in both Lemma~{\ref{teclem}}
and Lemma~{\ref{mainlem}} --- with any trace-preserving, positive linear map (i.e., positivity rather than
complete positivity is relevant therein).
\end{remark}

\begin{remark}
The function $\roq$ associated with the R\'enyi entropy is concave
for $0<q\le 1$ (see, e.g.,~{\cite{Principe}}, chapter~2), but not, in general,
for $q>1$; more precisely, it looses concavity for $q>\quz$, where the $\dime$-dependent
number $\quz$ is such that $1<\quz\le 1 + \ln(4)/\ln(\dime-1)$.
Hence, concavity could not be invoked, for $q>1$, in the part of the proof
of Lemma~{\ref{mainlem}} involving this kind of entropy. As we have seen, one can exploit
the \emph{Schur} concavity of $\req$, instead.
\end{remark}

\begin{remark} \label{symschu-bis}
It is worth observing, in connection with the previous remark, that
a concave function, defined on a convex symmetric domain in $\erren$
(e.g., on the probability simplex $\simplmo$),
is Schur concave if and only if it is symmetric; see~{\cite{Marshall}}, chapter~3,
C.2 (it is easy to check that the proof of this result holds for every domain of the
mentioned type). Hence, all the quantum entropies $\{\tseq\}_{q> 0}$, $\{\req\}_{q> 0}$,
$\{\nep\}_{1<\pap\le\infty}$ are actually Schur concave, in the sense specified in
sect.~{\ref{basic}}. E.g., recall that the quantities $\{\nep\}_{1<\pap\le\infty}$ are directly related
to the Schatten norms, which are, like every unitarily invariant norm, \emph{symmetric gauge functions}
of the singular values of their (matrix) argument; then, regarded as (symmetric) functions of the singular values,
they are Schur convex~{\cite{Marshall}}. Nevertheless, we think that distinguishing the
two cases --- the concave case and the Schur concave one --- in the proof of
Lemma~{\ref{mainlem}} is, although not necessary, interesting and instructive.
\end{remark}

We will also need a further observation.

\begin{lemma} \label{killem}
A quantum dynamical semigroup $\{\qds\}_{t\in\errep}$,
\begin{equation}
\qds = \exp(\gene\sei t) \fin ,
\end{equation}
is unital --- $\qds\tre \id=\id$, $t\ge 0$ --- if and only if
the generator $\gene$ kills the identity: $\gene\sei \id=0$.
\end{lemma}

\begin{proof}
Trivial.
\end{proof}

We are now ready to state and prove the main result of this section.

\begin{theorem} \label{mainth}
Let $\hh$ be a finite-dimensional complex Hilbert space, and let
$\{\qds\colon\bH\rightarrow\bH\}_{t\in\errep}$ be a
quantum dynamical semigroup with generator $\gene$.
Then, the following properties are equivalent:
\begin{enumerate}

\item
$\{\qds\}_{t\in\errep}$ is unital: $\qds\tre \id=\id$, $t\ge 0$;

\item
for every $t\ge 0$ and every $\hrho\in\stah$, $\qds\sei\hrho\prec\hrho$;

\item
$\{\qds\}_{t\in\errep}$ does not decrease the entropy $\entro$;

\item
the infinitesimal generator $\gene$ is of the form
\begin{equation} \label{genentro}
\gene\sei\opa = - \ima \big[\oph, \opa\big] +
\cpm\sei\opa - \frac{1}{2}
\left(\big(\cpm\sei \id)\tre\opa+\opa\sette\big(\cpm\sei \id\big)\right) ,
\end{equation}
for some completely positive map $\cpm$ such that
\begin{equation} \label{genentro-bis}
\cpm \sei \id = \cpm^\ast \id \fin .
\end{equation}

\end{enumerate}
\end{theorem}

\begin{proof}
By Lemma~{\ref{teclem}} the first two properties are equivalent.
By Proposition~{\ref{cohe}} and Lemma~{\ref{mainlem}}, the quantum dynamical
semigroup $\{\qds\}_{t\in\errep}$,
\begin{equation}
\qds = \exp(\gene\sei t) \fin ,
\end{equation}
does not decrease the entropy $\entro$ if and only if it is unital.
On the other hand, by Lemma~{\ref{killem}} $\{\qds\}_{t\in\errep}$ is unital if and only if
the generator $\gene$ kills the identity. Moreover, the condition $\gene\sei \id=0$
is equivalent to the fact that $\cpm \sei \id = \cpm^\ast \id$, where
$\cpm \colon \bH\rightarrow\bH$ is the completely positive map appearing in the
canonical form~{(\ref{forgene})} of the generator. The proof is complete.
\end{proof}

\begin{remark} \label{schano}
Suppose that a quantum dynamical semigroup $\{\qds\colon\bH\rightarrow\bH\}_{t\in\errep}$
(strictly) decreases some of the previously considered entropies (for some $t>0$ and some state).
Then, by Theorem~{\ref{mainth}} the same property holds for all other entropies, in particular
for $\nep$. It follows that, for every $1<\pap\le\infty$,
there exist some $t>0$ and some state $\hrho\in\stah$
such that
\begin{equation} \label{viol}
\|\qds\sei\hrho\norp > \|\hrho\norp \fin  .
\end{equation}
(Clearly, the subset of $\stah$ satisfying~{(\ref{viol})}
will always contain $\mms$.) Therefore, for some $t>0$,
\begin{equation}
\|\qds\norpp\defi\sup_{0\neq\opa\in\bH}
\big(\big\|\qds\sei\opa\big\norp / \big\|\opa\big\norp\big)> 1 \fin .
\end{equation}
Namely, $\{\qds\}_{t\in\errep}$ is not contractive w.r.t.\
the norm $\|\cdot\norpp$, for every $1<\pap\le\infty$
(whereas every quantum dynamical semigroup is contractive
w.r.t.\ the norm $\|\cdot\norpo$, see Remark~{\ref{contra}}).
\end{remark}

\begin{remark} \label{onlypos-bis}
Observe that Definition~{\ref{sedono}}, Proposition~{\ref{cohe}} and Lemma~{\ref{killem}}
extend immediately to semigroups of trace-preserving positive (not necessarily completely positive) maps.
As the reader may check, by this observation and by Remark~{\ref{onlypos}} the statement of
Theorem~{\ref{mainth}} remains valid --- with the only exception of the specific
form~{(\ref{genentro})}--{(\ref{genentro-bis})} of the generator $\gene$ --- if the quantum dynamical
semigroup $\{\qds\}_{t\in\errep}$ is replaced with a semigroup of operators whose members are assumed
to be trace-preserving, positive linear maps. Therefore, also the previous Remark~{\ref{schano}}
extends to this larger class of dynamical semigroups.
\end{remark}

\begin{remark} \label{rekss}
The completely positive map $\cpm$ in~{(\ref{genentro})} can be expressed in the
\emph{Kraus-Stinespring-Sudarshan form}
\begin{equation} \label{kssfo}
\cpm\sei\opa = \sum_{k=1}^{m}
\opkk\opa\sei\opkka \fin ,\ \ \ \opa\in\bH \fin ,
\end{equation}
and the condition $\cpm \sei \id = \cpm^\ast \id$ is expressed in terms of
the set of operators $\{\opko,\ldots,\opkm\}$ by the requirement that these
be \emph{jointly normal}, i.e.,
\begin{equation} \label{joinor}
\sum_{k=1}^{m} \opkk\tre\opkka = \sum_{k=1}^{m} \opkka\tre\opkk \fin .
\end{equation}
However, decomposition~{(\ref{kssfo})} is not unique, and one can further require that they be
\emph{traceless} (by a suitable redefinition of the Hamiltonian $\oph$) --- namely, that they live in
the orthogonal complement (w.r.t.\ the Hilbert-Schmidt scalar product) of the one-dimensional
subspace spanned by the identity --- and \emph{mutually orthogonal}:
\begin{equation}
\big\langle \opkj,\opkk\big\ranglehs :=
\tr\big(\opkja\cinque\opkk\big)= \kappa_j \tre\delta_{jk} \fin ,\
\ \ j,k =1,\dots,m \ \ (m\le\dime^2 -1) \fin , \ \ \ \kappa_j >0 \fin .
\end{equation}
Implementing these requirements one obtains the so-called  \emph{diagonal form} of
the infinitesimal generator $\gene$~{\cite{Breuer,Aniello1}}.
\end{remark}

We will now derive a few remarkable consequences of the main theorem.

\begin{corollary} \label{coradj}
A quantum dynamical semigroup $\{\qds\colon\bH\rightarrow\bH\}_{t\in\errep}$
does not decrease the quantum entropy $\entro$ if and only if the adjoint
semigroup is a quantum dynamical semigroup too.
\end{corollary}

\begin{corollary} \label{corpur}
A quantum dynamical semigroup $\{\qds\colon\bH\rightarrow\bH\}_{t\in\errep}$
does not decrease the quantum entropy $\entro$ if and only if it does not
increase the purity.
\end{corollary}

\begin{proof}
Recall that $\tset(\hrho)=1-\tr(\hrhot)$.
\end{proof}

\begin{corollary} \label{corgen}
A quantum dynamical semigroup $\{\qds\colon\bH\rightarrow\bH\}_{t\in\errep}$
does not decrease the quantum entropy $\entro$ if and only if
its infinitesimal generator is of the form
\begin{equation} \label{forgeneg}
\gene\sei\opa = - \ima \big[\oph, \opa\big] +
\sum_{k=1}^{\dime^2-1} \gamma_k\Big(\opfk\opa\sei\opfka - \frac{1}{2}
\big(\opfka\opfk\opa+ \opa\sei\opfka\opfk\big)\Big),
\end{equation}
where $\oph$ is a trace-less selfadjoint operator,
$\gamma_1\ge0,\ldots, \gamma_{\dime^2-1}\ge 0$, and $\opfo,\ldots, \opfd$
are traceless operators such that
\begin{equation} \label{orthog}
\big\langle
\opf_{j}^{\phantom{\ast}},\opf_k^{\phantom{\ast}}\big\ranglehs
=\delta_{jk} \fin ,\ \ \ j,k =1,\dots,\dime^2-1 \fin ,
\end{equation}
and
\begin{equation}
\sum_{k=1}^{\dime^2-1} \gamma_k \tre \opfk\opfka = \sum_{k=1}^{\dime^2-1} \gamma_k\tre \opfka\opfk \fin .
\end{equation}
\end{corollary}

\begin{proof}
This result is a straightforward consequence of Theorem~{\ref{mainth}} and of the `diagonal form' of
the infinitesimal generator $\gene$; see Remark~{\ref{rekss}}.
\end{proof}

Let us now consider some examples.

\begin{example}
Consider the quantum dynamical semigroup $\{\qds\colon\bH\rightarrow\bH\}_{t\in\errep}$, with
\begin{equation}
\qds\sei\opa\defi\eee^{-\lambda t} \sum_{n=0}^\infty \frac{(\lambda\tre t)^n}{n!}\otto \unopn\opa\sei\unopan ,
\ \ \ \lambda>0 \fin , \ \ \opu\in\unihh \fin .
\end{equation}
Clearly, this is unital and the associated infinitesimal generator $\gene$ is given by
\begin{equation}
\gene\sei\opa = \lambda\tre (\unop\opa\sei\unopa - \opa) \fin .
\end{equation}
Thus, $\gene$ kills the identity. By Theorem~{\ref{mainth}}, $\{\qds\}_{t\in\errep}$ does not
decrease the entropy $\entro$.
\end{example}

\begin{example} \label{exgesa}
Consider the quantum dynamical semigroup in $\bH$ whose infinitesimal generator is of the form
\begin{equation} \label{forgesa}
\gene\sei\opa = - \ima \big[\oph, \opa\big] +
\sum_{k=1}^{\dime^2-1} \Big(\oplk\opa\sei\oplk -\frac{1}{2}
\big(\oplks\tre\opa+ \opa\sei\oplks\big)\Big),
\end{equation}
where $\oph$ and $\oplo,\ldots, \opld$ are selfadjoint operators. Thus, once again $\gene$ kills the identity, and then
the associated quantum dynamical semigroup does not decrease the entropy $\entro$.
\end{example}

\begin{example} \label{twisem}
The two previous examples can be regarded as special cases of the following.
Let $\{\mut\}_{t\in\errep}$ be a convolution semigroup of probability measures~{\cite{Heyer}}
on a locally compact (second countable, Hausdorff topological) group $G$,
and let $\repr$ be a continuous unitary representation of this group in $\hh$. Then,
$\{\qds\colon\bH\rightarrow\bH\}_{t\in\errep}$, with
\begin{equation} \label{twise}
\qds\sei\opa\defi \int_{G} \mtre \repr(g) \opa \sei \repr(g)^\ast \; \de\mut(g) \fin ,
\end{equation}
turns out to be a quantum dynamical semigroup, a so-called
\emph{twirling semigroup}~{\cite{Aniello1,Aniello2,Aniello3}}.
It can be shown that, in order to achieve a \emph{generic} twirling semigroup,
one can always set $G=\sunig(\dime)$ and choose $\repr$ as the defining representation
(with the obvious identification of $\unig(\dime)$ with $\unig(\hh)$), where $\{\mut\}_{t\in\errep}$ ranges over
the whole set of convolution semigroups of probability measures on $\sunig(\dime)$; see~{\cite{Aniello1}}.
It is clear that $\{\qds\}_{t\in\errep}$ is also unital; hence, it does not decrease the entropy $\entro$.
The generator of this quantum dynamical semigroup is of the form
\begin{equation} \label{fogeneg}
\gene\sei\opa = - \ima \big[\oph, \opa\big] + \sum_{k=1}^{\dime^2-1}
\gamma_k\Big(\opl_k^{\phantom{\ast}}\opa\sei\opl_{k}^{\phantom{\ast}}-\frac{1}{2}
\big(\opl_k^2\tre\opa+ \opa\sei\opl_k^2\big)\Big) \msei +
\gamma_0\big(\randu - \ide\big)\tre \opa \fin ,\ \ \ \opa\in\bH \fin ,
\end{equation}
where $\oph$ is a trace-less selfadjoint operator,
$\oplo,\ldots, \opld$
are trace-less selfadjoint operators such that
\begin{equation} \label{orthogsa}
\big\langle
\opl_{j}^{\phantom{\ast}},\opl_k^{\phantom{\ast}}\big\ranglehs
=\delta_{jk} \fin ,\ \ \ j,k =1,\dots,\dime^2-1 \fin ,
\end{equation}
$\gamma_0,\ldots,\gamma_{\dime^2-1}$ are non-negative numbers and
$\randu$ is a \emph{random unitary map} acting in $\bH$; namely,
it admits a decomposition of the
form
\begin{equation} \label{defirandu}
\randu\sette \opa= \sum_{j=1}^{\den} p_j\sei \uj\opa\otto\ujast \fin ,
\end{equation}
where  $\big\{\uj\big\}_{j=1}^{\den}$ is a set of unitary operators in $\hh$
and $\{p_j\}_{j=1}^{\den}\subset\erreps$ is a probability
distribution. Note that the operators $\uo,\ldots,\un,\oplo,\ldots, \opld$
are normal (hence, jointly normal). Conversely, every operator $\gene$ of the
form~{(\ref{fogeneg})} is the infinitesimal generators of a twirling semigroup
of the form~{(\ref{twise})}.

At this point, it is worth observing that the class of all twirling semigroups
of the form~{(\ref{twise})} turns out to coincide with the class of
\emph{random unitary semigroups} acting in $\bH$~{\cite{Aniello1,Aniello3}}, namely,
of those unital quantum dynamical semigroups whose members are random unitary maps.
On the other hand, it is known~{\cite{Landau}} that every
unital quantum channel acting in $\bH$, for $\dime=\dim(\hh)=2$, is actually
a random unitary map (but for $\dime>2$ this is no longer true).
By these facts and by Theorem~{\ref{mainth}}, we see
that, for $\dime=2$, the class of twirling semigroups actually \emph{coincides}
with the class of quantum dynamical semigroups that do not decrease
the entropy $\entro$.
\end{example}

\begin{example}
Consider a qubit generator of the form
\begin{equation}\label{qubit}
\gene \sei \opa = \gamma_1 \msei\left(\hat{\sigma}_- \sei \opa \sei \hat{\sigma}_+
-\frac{1}{2} \left(\hat{\sigma}_+\sei\hat{\sigma}_- \sei\opa
+ \opa\sei \hat{\sigma}_+\sei\hat{\sigma}_-\right) \right) \msei
+ \gamma_2 \msei\left( \hat{\sigma}_+\sei \opa\sei \hat{\sigma}_-
- \frac{1}{2} \left( \hat{\sigma}_-\sei\hat{\sigma}_+\sei \opa +
\opa\sei \hat{\sigma}_-\sei\hat{\sigma}_+\right) \right) ,
\end{equation}
where $\hat{\sigma}_\pm = \frac{1}{2}(\hat{\sigma}_1 \pm i \hat{\sigma}_2)$
are the standard raising and lowering qubit operators. One easily finds that
$\gene\sei \id = (\gamma_1-\gamma_2)\sei\hat{\sigma}_3$; hence, for $\gamma_1 \neq \gamma_2$,
$\gene\sei \id\neq 0$ and the corresponding temporal evolution $\exp(\gene\sei t)$
may decrease the entropy of some state.
Indeed, if the the entropy $\entro(\hrho(0))$ of the initial state $\hrho(0)$ is strictly larger
than the entropy of the  asymptotic state $\asy$ determined by the probability vector
\begin{equation}
\left(\po = \frac{\gamma_1}{\gamma_1 + \gamma_2},\tre  \pt = \frac{\gamma_2}{\gamma_1 + \gamma_2} \right),
\end{equation}
then, for some $t>0$, $t\mapsto\entro(\hrho(t))$ must be decreasing.
If $\gamma_1=\gamma_2=\gamma$, then (\ref{qubit}) reduces to
\begin{equation}
\gene\sei \opa =  \gamma \left(\frac{1}{2} \sei\hat{\sigma}_1 \tre \opa \sei \hat{\sigma}_1
+ \frac{1}{2} \sei \hat{\sigma}_2 \tre \opa \sei \hat{\sigma}_2 - \opa \right) ,
\end{equation}
which generates a random unitary evolution, and $\asy$ is in this case the maximally mixed state.
\end{example}

\begin{example} Let  $\supro\colon\bH\rightarrow\bH$ be a completely positive, trace-preserving \emph{projection}
--- i.e., $\supro^2 = \supro$ --- and set  $\suprop \defi \ide  - \supro$. Consider the generator $\gene$,
\begin{equation} \label{generator}
\gene \sei\opa \defi \gamma \sei\supro\sei\opa - \frac{\gamma}{2}
\left(\big(\supro^\ast \id\big)\tre\opa+\opa\sette\big(\supro^\ast \id\big)\right) =
- \gamma \sei \suprop \opa \fin ,
\end{equation}
with $\gamma >0$. We stress that $\supro^\ast$ is unital, because the completely positive
map $\supro$ is assumed to be trace-preserving. One easily finds for the evolution
\begin{equation}
\qds = \supro + \eee^{-\gamma t}\sei \suprop .
\end{equation}
The quantum dynamical semigroup $\{\qds\}_{t\in\errep}$ is unital if and only $\suprop$ kills the identity;
equivalently, if and only if $\supro$ itself is unital. Therefore, $\{\qds\}_{t\in\errep}$ does not decrease
the entropy $\entro$ if and only if the completely positive, trace-preserving projection $\supro$ is also
unital.

Suppose, e.g., that $\suproz\tre\opa = \sum_k \orthk \opa \sei\orthk$, where $\{\orthk\}$
are mutually orthogonal (selfadjoint) projectors in the Hilbert space $\hh$ such that $\sum_k \orthk = \id$.
Then, $\suproz$ is a completely positive, trace-preserving, unital projection.
Note, in this regard, that there are projections in $\bH$ that are not unital. Consider, e.g.,
the case where
\begin{equation}
\supro\sei\opa = \tr(\opa)\sei\hrho \fin ,
\end{equation}
$\hrho$ being a \emph{fixed} density operator in $\stah$.
This projection is unital only if $\hrho$ is maximally mixed.
\end{example}

It is worth condensing the discussion concluding Example~{\ref{twisem}} as follows.

\begin{corollary} \label{satur}
For $\dime=\dim(\hh)=2$, a quantum dynamical semigroup $\{\qds\colon\bH\rightarrow\bH\}_{t\in\errep}$
does not decrease the entropy $\entro$ if and only if it is a twirling semigroup. In the integral expression
of such a semigroup of operators, one can always set $G=\sunig(2)$ and choose $\repr$ as the defining representation.
\end{corollary}

\section{Further results: relaxing the complete positivity}
\label{further}

Interestingly, as previously observed --- see Remark~{\ref{onlypos-bis}} ---
the statement of Theorem~{\ref{mainth}} turns out to be partially valid
if the quantum dynamical semigroup $\{\qds\}_{t\in\errep}$ is replaced with a semigroup
of operators $\{\pds\}_{t\in\errep}$ whose members are assumed to be trace-preserving,
positive linear maps (thus, the condition of \emph{complete} positivity being `relaxed').
Precisely, in this more general context we have the following result.

\begin{theorem} \label{possem}
Let $\{\pds\colon\bH\rightarrow\bH\}_{t\in\errep}$ be a (continuous) semigroup of linear maps,
with infinitesimal generator $\gene$. Then, the following properties are equivalent:
\begin{enumerate}

\item
$\{\pds\}_{t\in\errep}$ is positive, trace-preserving and unital;

\item
$\{\pds\}_{t\in\errep}$ is positive, trace-preserving and,
for every $t\ge 0$ and every $\hrho\in\stah$, $\pds\sei\hrho\prec\hrho$;

\item
$\{\pds\}_{t\in\errep}$ is positive, trace-preserving and does not decrease
the entropy $\entro$;

\item
for every set $\big\{\proj\big\}$ of mutually orthogonal rank-one (selfadjoint) projectors in the Hilbert space
$\hh$, such that $\sum_j \proj = \id$,
\begin{equation} \label{posuno}
\tr\big(\proj\big(\gene\sei\prok\big)\big)\ge 0 \fin , \ \ \mbox{for $j\neq k$} \fin ,
\end{equation}
and
\begin{equation} \label{posdue}
\sum_{j=1}^{\dime}\tr\big(\proj\big(\gene\sei\prok\big)\big) = 0
= \sum_{k=1}^{\dime}\tr\big(\proj\big(\gene\sei\prok\big)\big) \fin ;
\end{equation}

\item
for every pair of mutually orthogonal (selfadjoint) projectors $\orthp,\orthq\in\bH$,
\begin{equation} \label{postre}
\big\langle\orthp,\gene\sei\orthq\big\ranglehs\ge 0 \ \ \mbox{and} \ \
\big\langle \id , \gene\sei\orthp\big\ranglehs = 0
= \big\langle \orthp , \gene\sei\id \big\ranglehs \fin .
\end{equation}

\end{enumerate}

\end{theorem}

\begin{proof}
For proving the equivalence of the first three properties above one proceeds as in
the proof of Theorem~{\ref{mainth}}.
Inequality~{(\ref{posuno})} and the first of equalities~{(\ref{posdue})} are known
(see~{\cite{Gorini}}, Theorem~{2.1}) to form a necessary and sufficient condition
for $\gene$ to be the generator of a semigroup of trace-preserving,
positive linear maps (precisely, condition~{(\ref{posuno})} is equivalent to positivity,
whereas the other condition amounts to preservation of the trace).
Moreover, by the arbitrariness of the set $\big\{\proj\big\}$, the second of
equalities~{(\ref{posdue})} is equivalent to $\gene\sei \id=\sum_k \big(\gene\sei\prok\big)=0$
(note that, by a standard polarization argument, the fact that $\langle\psi,\opa\tre\psi\rangle=0$,
for every $\psi\in\hh$, $\|\psi\|=1$, implies that $\opa=0$).
Hence, by these observations and by Lemma~{\ref{killem}} (which obviously holds for semigroups
of positive maps as well) the fourth property is equivalent to the
previous three. Finally, the equivalence between conditions~{(\ref{posuno})}--{(\ref{posdue})}
and~{(\ref{postre})} is easily seen.
\end{proof}

From the fourth of the equivalent properties listed in Theorem~{\ref{mainth}} and the fifth in
Theorem~{\ref{possem}} one immediately derives the following fact.

\begin{corollary} \label{concon}
The generators of the semigroups of trace-preserving, positive linear maps acting in $\bH$ that are
also unital --- equivalently, that do not decrease the entropy $\entro$ ---
form a convex cone $\ptu$ in the space of linear maps in $\bH$,
which contains another convex cone $\cptu$ formed by the generators
of semigroups of unital, trace-preserving, completely positive linear maps.
\end{corollary}

A connection with Example~{\ref{twisem}} and with the related Corollary~{\ref{satur}}
is established by the following fact.

\begin{proposition} \label{genertwi}
Let $\{\pds\colon\bH\rightarrow\bH\}_{t\in\errep}$ be a family of linear maps.
$\{\pds\}_{t\in\errep}$ is a semigroup of positive, trace-preserving and unital
maps if and only if
\begin{equation} \label{intefo}
\pds\sei\opa = \int_G
 \mtre \repr(g) \opa \sei \repr(g)^\ast \; \de\sut(g) \fin ,
\end{equation}
where $(\repr,\{\sut\}_{t\in\errep})$ is a pair formed by a continuous unitary representation
$\repr\colon G\rightarrow\unihh$ of a locally compact group $G$ and by a family $\{\sut\}_{t\in\errep}$
of finite, signed Borel measures on $G$ satisfying the following conditions:
\begin{description}

\item[\tt (M1)]
$\suz=\delta$ (Dirac measure at the identity of $G$) and $\sut(G)=1$, for every $t\in\erreps$;

\item[\tt (M2)]
for some (hence, for every) orthonormal basis $\oba\equiv\{\psi_k\}_{k=1}^{\dime}$ in $\hh$,
\begin{equation} \label{inteid}
\lim_{t\downarrow 0} \int_G \vjklm \; \de\sut(g) = \delta_{jk}\cinque \delta_{lm}
\end{equation}
and
\begin{equation} \label{intese}
\int_G \vjklm \; \de(\sus\convo\sut)(g) = \int_G \vjklm \; \de\sust(g) \fin , \ \ \
\forall\quattro s,t\in \erreps \fin ,
\end{equation}
where $\vjklm\defi\langle \psi_j,\repr(g)\tre\psi_k\rangle\langle \repr(g)\tre\psi_l,\psi_m\rangle$
and $\sus\convo\sut$ is the convolution of $\sus$ with $\sut$;

\item[\tt (M3)]
for every orthonormal basis $\oba\equiv\{\psi_k\}_{k=1}^{\dime}$ in $\hh$,
\begin{equation} \label{intepos}
\lim_{t\downarrow 0} \tre t^{-1} \int_{G} \vjkkj \;\de\sut(g)\ge 0 \fin , \ \ \ j\neq k \fin .
\end{equation}

\end{description}

In particular, one can always set $G=\sunig(\dime)$ and choose $\repr$ as the defining representation.

Moreover, a semigroup of trace-preserving positive maps $\{\pds\colon\bH\rightarrow\bH\}_{t\in\errep}$
does not decrease the entropy $\entro$ if and only if it is of the form~{(\ref{intefo})}, for
some locally compact group $G$, some continuous unitary representation $\repr$ of $G$ in $\hh$
and some family $\{\sut\}_{t\in\errep}$ of finite, signed Borel measures on $G$.
\end{proposition}

\begin{proof}
If a family of linear maps $\{\pds\colon\bH\rightarrow\bH\}_{t\in\errep}$ is of
the form~{(\ref{intefo})} --- for some locally compact group $G$, some continuous
unitary representation $\repr$ of $G$ in $\hh$ and some family $\{\sut\}_{t\in\errep}$
of finite, signed Borel measures on $G$ such that $\sut(G)=1$ --- then these maps
are trace-preserving and unital. If, moreover, $\suz=\delta$, where $\delta$ is the
Dirac measure at the identity of $G$, and both the
conditions~{(\ref{inteid})} and~{(\ref{intese})} in~{\tt (M2)} are satisfied,
it is easy to see that $\{\pds\}_{t\in\errep}$ is a continuous semigroup of operators
(observe that the partial isometries $|\psi_j\rangle\langle\psi_k|$, $j,k\in\{1,\ldots,\dime\}$,
form a basis in $\bH$). If in addition the family of measures $\{\sut\}_{t\in\errep}$  satisfies
condition~{(\ref{intepos})}, we have that, for every orthonormal basis $\{\psi_k\}_{k=1}^{\dime}$
in $\hh$, and for $j\neq k$,
\begin{equation}
\tr\big(\proj\big(\gene\sei\prok\big)\big) = \lim_{t\downarrow 0} \tre t^{-1} \sei
\tr\big(\proj\big((\pds -\ide)\tre\prok\big)\big) = \lim_{t\downarrow 0} \tre t^{-1}
\int_{G} |\langle \psi_j,\repr(g)\tre\psi_k\rangle|^2 \;\de\sut(g)\ge 0 \fin ,
\end{equation}
where $\gene$ is the infinitesimal generator of the semigroup of operators
and $\prok=|\psi_k\rangle\langle\psi_k|$. Thus, recalling that~{(\ref{posuno})}
is the condition expressing the positivity of the semigroup of linear maps
(see the proof of Theorem~{\ref{possem}}), we conclude that $\{\pds\}_{t\in\errep}$
is a semigroup of unital, trace-preserving positive linear maps.

Conversely, let $\{\pds\colon\bH\rightarrow\bH\}_{t\in\errep}$ be a semigroup of this kind. Then,
for every $t\in\errep$, since the linear map $\pds$ is positive, trace-preserving and unital,
it admits a decomposition of the form (see~{\cite{Mendl}}, Theorem~{1}):
\begin{equation}
\pds\sei\opa = \sum_m \wem \sei \um \opa\sei \umast , \ \ \ \{\wem\}\subset\erre, \ \
\sum_m \wem = 1 \fin ,
\end{equation}
where $\big\{\um\big\}$ is a (finite) set of unitary operators in $\hh$. Hence,
for $G=\sunig(\dime)$, choosing $\repr$ as the defining representation and for a suitable
finite, signed Borel measure $\sut$ on $G$ (a linear combination of Dirac measures),
$\pds$ is of the form~{(\ref{intefo})}. Note that $\sut(G)= \sum_m \wem = 1$ and,
arguing essentially as in the first part of the proof, one concludes that by the fact that
$\{\pds\}_{t\in\errep}$ is a continuous semigroup of positive maps conditions~{(\ref{inteid})},
{(\ref{intese})} and~{(\ref{intepos})} must be satisfied too.

At this point, taking into account the equivalence between the first and the third of the
properties listed in Theorem~{\ref{possem}}, the last assertion of the statement is also clear.
\end{proof}

It is worth observing that, as it should be clear from the previous proof, the family of operators
$\{\pds\}_{t\in\errep}$ defined by~{(\ref{intefo})} is a semigroup of operators if and only if
the conditions $\suz=\delta$ and~{\tt (M2)} are satisfied; moreover, if these conditions are verified
then the limit in~{\tt (M3)} exists. Also note that {\tt (M1)}, {\tt (M2)} and~{\tt (M3)}
are satisfied by every convolution semigroup of probability measures $\{\mut\equiv\sut\}_{t\in\errep}$.

As the reader will easily check, conditions~{(\ref{postre})} are satisfied,
in particular, if the generator $\gene$ is of the
form~{(\ref{genentro})}--{(\ref{genentro-bis})}, where now the linear map $\cpm$
is only assumed to be positive (rather than completely positive). Not much is known
about the general structure of positive maps, but a remarkable case is that of
\emph{decomposable} ones, i.e., of those positive linear maps of the form
\begin{equation} \label{deco}
\bH\ni\opa\mapsto \big(\cpmo\tre\opa+\cpmt\big(\coco\tre\opa^\ast\mtre\coco\big)\big),
\end{equation}
where $\cpmo$, $\cpmt$ are completely positive maps and $\coco$ is a complex
conjugation in $\hh$ (a selfadjoint antiunitary operator); otherwise stated, the linear map
$\opa\mapsto\traspo\big(\opa\big)\defi\coco\tre\opa^\ast\mtre\coco$ is a transposition.
For $\dime=\dim(\hh)=2$, the expression~{(\ref{deco})} actually gives the
\emph{general} form of a positive map~{\cite{Stormer}}, as the sum of a completely
positive map and a completely \emph{co-}positive map. Clearly, in this
case the condition $\cpm \sei \id = \cpm^\ast \id$ amounts to imposing that
\begin{equation}
\cpmo \sei \id  + \cpmt \sei \id = \cpmoa \sei \id  + \coco\tre\big(\cpmta\sei\id\big)\coco \fin .
\end{equation}

Thus, by Theorem~{\ref{possem}} a generator $\gene$, associated with a positive map $\cpm$
satisfying the condition $\cpm\sei\id=\cpm^\ast\id$, gives rise to a semigroup of unital,
trace-preserving, positive linear maps. But it can be shown with examples that the reverse
implication does not hold; not even in the simplest case where $\dim(\hh)=2$ --- see Example~{\ref{nonpos}}
\emph{infra} --- and hence the positive maps in $\bH$ coincide with the decomposable ones.

To exhibit the general form --- in the qubit case --- of the generator
of a semigroup of trace-preserving positive maps, it will be convenient to
adopt a Bloch ball approach. Let us then choose the following basis in $\bH$ (orthogonal w.r.t.\
the Hilbert-Schmidt product):
\begin{equation}
\spz\equiv\frac{1}{2}\sei\siz=\frac{1}{2}\sei\id \fin, \  \ \spb\equiv\left(\spo\equiv\frac{1}{2}\sei\sio\tre ,
\ \spt\equiv\frac{1}{2}\sei\sit \tre , \ \sptr\equiv\frac{1}{2}\sei\sitr\right) ,
\end{equation}
where $\siz,\ldots,\sitr$  (in matrix form, w.r.t.\ some orthonormal basis in $\hh$)
are the standard Pauli matrices.

\begin{remark}
Of course, $\spo,\spt,\sptr$ can be defined as `abstract' selfadjoint operators satisfying
the condition
\begin{equation} \label{absdef}
\spj\tre\spk= \frac{1}{4}\sei\delta_{jk}\sei\id + \frac{\ima}{2}\sum_{l=1}^3
\epsilon_{jkl} \sei\spl \fin , \ \ \ j,k\in\{1,2,3\} \fin .
\end{equation}
\end{remark}

For every pair of operators $\opa=\az\tre\spz+ \aaa\cdot\spb$
($\aaa\cdot\spb\equiv\ao\tre\spo+\at\tre\spt+\atr\tre\sptr$),
$\opb=\bz\tre\spz+ \bbb\cdot\spb$ in $\hh$, we have:
\begin{equation} \label{hspro}
\big\langle \opa,\opb\big\ranglehs = \big(\az\tre\bz+\langle\aaa,\bbb\ranglec\big)/2 \fin .
\end{equation}
The operator $\opa$ is positive if and only if $\az=\tr\big(\opa\big)\ge 0$, $\ao,\at,\atr\in \erre$ and
\begin{equation} \label{posco}
\az\ge\|\aaa\|=\sqrt{\tre\langle\aaa,\aaa\ranglec\msei} \ ;
\end{equation}
in particular, for a density operator $\hrho$,
\begin{equation} \label{bloba}
\hrho = \spz + \rrb\cdot\spb,\ \ \ \ru,\rt,\rtr\in\erre \fin ,
\ \ \ \|\rrb\|\le 1 \fin ,
\end{equation}
where the state $\hrho$ is pure if and only if the point $\rrb$ lies on the surface
of the Bloch ball (i.e., on the Bloch sphere $\|\rrb\|= 1$).
Accordingly, we will also consider the associated matrix representation
of a semigroup of operators $\{\pds\colon\bH\rightarrow\bH\}_{t\in\errep}$ and of its
generator $\gene$:
\begin{equation} \label{marep}
\mase_{jk}\defi \tr\big(\spj\big(\pds\sei\spk\big)\big), \ \ \
\mage_{jk}\defi \tr\big(\spj\big(\gene\sei\spk\big)\big) ,
\ \ \ j,k\in\{0,\ldots,3\} \fin .
\end{equation}

\begin{proposition} \label{posetwo}
The general form of a (continuous) semigroup of unital, trace-preserving, positive linear maps
$\{\pds\colon\bH\rightarrow\bH\}_{t\in\errep}$ --- for $\dime=\dim(\hh)=2$ --- is given, in the
matrix representation~{(\ref{marep})}, by
\begin{equation} \label{gefor}
\mase=\exp(\mage\tre t) = \begin{pmatrix} 1 & \zzz \\ \zzz & \eee^{\rema t}\end{pmatrix},
\ \ \ \mage = \begin{pmatrix} 0 & \zzz \\ \zzz & \rema \end{pmatrix},
\end{equation}
where $\rema$ is a $3\times 3$ real matrix such that $\rema + \remat\le 0$.
Therefore, a semigroup of trace-preserving, positive linear maps does not decrease
the entropy $\entro$ if and only if its matrix representation is of the given form.
\end{proposition}

\begin{proof}
Since $\{\pds\}_{t\in\errep}$ is positive (hence, adjoint-preserving),
trace-preserving and unital, its matrix representation
w.r.t.\ the orthogonal basis $(\spz,\ldots,\sptr)$ in $\bH$ is of the form~{(\ref{gefor})}, so that
\begin{equation}
\exp(\gene\sei t) \sei \opa = \az\tre\spz + \big(\eee^{\rema t}\tre\aaa\big)\cdot\spb  \fin ,
\end{equation}
where $\rema$ is a real matrix. We know, moreover, that $\{\pds\}_{t\in\errep}$
does not decrease the entropy $\nep$; in particular, for $\pap=2$. Therefore, for every $t>0$ and
$\hrho\in\stah$, we have that $\|\pds\sei\hrho\norhs\le\|\hrho\norhs$ and hence,
taking into account~{(\ref{hspro})},
\begin{equation}
1+\big\langle\eee^{\rema t}\tre\rrb,\eee^{\rema t}\tre\rrb\big\rangler =
2\sei\langle\pds\sei\hrho,\pds\sei\hrho\ranglehs\le 2\sei\langle\hrho,\hrho\ranglehs
=1+\langle\rrb,\rrb\rangler \fin ,
\end{equation}
where $\rrb$ is the vector in $\rcub$ associated with $\hrho$, see~{(\ref{bloba})}.
It follows that for every $\rrb$ in the Bloch ball --- thus, by linearity, for every $\rrb\in\rcub$ ---
\begin{equation} \label{lanese}
\big\langle\rrb,\big(\rema + \remat\big)\tre\rrb\big\rangler =
\dert\sei\big\langle\eee^{\rema t}\tre\rrb,\eee^{\rema t}
\tre\rrb\big\rangler\sei\Big\vert_{t=0}\le 0 \fin .
\end{equation}
Since the real matrix $\rema + \remat$ is symmetric, relation~{(\ref{lanese})} means
that $\rema + \remat\le 0$.

Conversely, if the matrix representation of the semigroup of operators
$\{\pds\}_{t\in\errep}$ in $\bH$ is of the form~{(\ref{gefor})} --- where $\rema$ is a real matrix
satisfying the condition $\rema + \remat\le 0$ --- then the semigroup is unital, trace-preserving and,
for every $\rrb\in\rcub$ and $t\ge 0$,
\begin{equation}
0\ge \big\langle\eee^{\rema t}\tre\rrb,\big(\rema +
\remat\big)\cinque\eee^{\rema t}\tre\rrb\big\rangler =
\dert\sei\big\langle\eee^{\rema t}\tre\rrb,\eee^{\rema t}
\tre\rrb\big\rangler \fin  .
\end{equation}
Thus, the function $\errep\ni t\mapsto\big\langle\eee^{\rema t}\tre\rrb,\eee^{\rema t}\tre\rrb\big\rangler$
is not increasing and, by condition~{(\ref{posco})}, $\{\pds\}_{t\in\errep}$ is also positive.
\end{proof}

\begin{example} \label{nonpos}
Consider the case of a qubit generator $\gene$ is of the form
\begin{equation}
\gene\sei\opa=\sum_{k=1}^3 (\gak-\gamma)\tre \big\langle\spk,\opa\big\ranglehs\otto\spk \fin ,
\ \ \ \gao,\gat,\gatr\in\erre \fin , \ \ \ \gamma\equiv\gao+\gat+\gatr \fin .
\end{equation}
Note that $\gene$ is adjoint-preserving, and $\gene=\gene^\ast$;
i.e., $\big\langle \gene\sei\opa ,\opb \big\ranglehs=\big\langle \opa , \gene\sei\opb \big\ranglehs$.
It is also clear that $\gene\sei\id=0 =\gene^\ast\id$,
and hence $\gene$ generates a unital, trace-preserving and adjoint-preserving
dynamics. Moreover, the associated real matrix $\mage$ --- see~{(\ref{marep})}--{(\ref{gefor})} ---
is characterized by the $3\times 3$ symmetric sub-matrix $\rema$ of the form
\begin{equation}
\rema=-\frac{1}{2}\sei\diag\big(\gat+\gatr,\gao+\gatr,\gao+\gat\big) .
\end{equation}
Imposing the constraints
\begin{equation} \label{constra}
\gao+\gat\ge 0\fin ,\ \ \gao+\gatr\ge 0\fin , \ \ \gat+\gatr\ge 0 \fin ,
\end{equation}
amounts to assuming that $\rema$ is a negative semidefinite matrix.
Therefore, according to Corollary~{\ref{posetwo}}, $\gene$ is the generator
of a semigroup of unital, trace-preserving, \emph{positive} linear maps if and
only if the inequalities~{(\ref{constra})} hold.
On can easily check that this generator can be expressed in the canonical form
\begin{equation}
\gene\sei\opa = \cpm\sei\opa - \frac{1}{2}
\left(\big(\cpm\sei \id)\tre\opa+\opa\sette\big(\cpm\sei \id\big)\right) ,
\ \ \ \mbox{with} \ \cpm \sei \id = \cpm^\ast \id \fin ,
\end{equation}
where, however, in general the linear map $\cpm$ is neither completely positive nor positive. Indeed,
it turns out that
\begin{equation}
\cpm\sei\opa = \sum_{k=1}^3 \gak\sei \spk\tre\opa\sei\spk =
\frac{1}{4}\sei\gamma\sei\tr\big(\opa\big)\tre\spz +
\frac{1}{2}\sum_{k=1}^3 (2\gak -\gamma)\tre
\big\langle\spk,\opa\big\ranglehs\otto\spk \fin .
\end{equation}
Then, $\cpm$ is completely positive if and only if $\gao,\gat,\gatr\ge 0$.
Moreover, given the pair of mutually orthogonal rank-one projectors
$\spz+\sptr$ and $\spz-\sptr$, we find that
\begin{equation}
\cpm\sei(\spz+\sptr)= \frac{1}{4}\sei\gatr \big(\spz+\sptr\big) + \frac{1}{4}\sei(\gao+\gat) \big(\spz-\sptr\big).
\end{equation}
Therefore, it is clear that, in general, the map $\cpm$ is not even positive (e.g., for $\gao\ge\gat>0$
and $-\gat\le\gatr<0$).
\end{example}

The preceding example illustrates a simple subclass of the whole class of qubit generators of semigroups of
trace-preserving positive maps that do not decrease the entropy $\entro$.

\begin{proposition} \label{furpro}
For $\dim(\hh)=2$, the general form of the generator of a semigroup of unital, trace-preserving, positive linear maps
--- i.e., of a semigroup of trace-preserving, positive linear maps that does not decrease $\entro$ --- is
given by
\begin{equation} \label{geppo}
\gene\sei\opa = - \ima \sum_{j=1}^3 \hj \tre\big[\spj, \opa\big] +
\sum_{j,k=1}^3 \kajk\left(\spj\tre\opa\sei\spk-\frac{1}{2}\big(\spj\sei\spk\tre\opa+\opa\sei\spj\sei\spk\big)\right),
\end{equation}
where $\hon,\htw,\htr\in\erre$, the $3\times 3$ matrix $\copoma\defi(\kajk)$ is such
that $\kajk=\kakj\in\erre$ and --- setting
\begin{equation}
\ko\equiv\kat+\katr \fin , \ \kt\equiv\kao+\katr \fin , \ \ktr\equiv\kao+\kat \fin ,
\ a\equiv-\kattr \fin , \ b\equiv-\kaotr \fin , \ c\equiv-\kaot
\end{equation}
--- the associated symmetric real matrix
\begin{equation} \label{prema}
\mathscr{P}\defi
\begin{pmatrix}
\ko & c & b \\
c & \kt & a \\
b &  a & \ktr
\end{pmatrix}
\end{equation}
is positive semidefinite. Thus, such a generator $\gene$ is selfadjoint --- $\gene=\genea$
--- if and only if $\hon=\htw=\htr=0$.

The semigroup of linear maps associated with $\gene$ is, in particular,
completely positive if and only if the symmetric real matrix $\copoma$ is positive semidefinite.
\end{proposition}

\begin{proof}
Exploiting a classical result (see~{\cite{Gorini}}, Lemma~{2.3}), one can easily check that,
for $\dim(\hh)=2$, the generator of a semigroup of trace-preserving and adjoint-preserving linear maps is of the form
\begin{equation} \label{geppo-bis}
\gene\sei\opa = - \ima \sum_{j=1}^3 \hj \tre\big[\spj, \opa\big] +
\sum_{j,k=1}^3 \kajk\left(\spj\tre\opa\sei\spk-\frac{1}{2}\big(\spk\sei\spj\tre\opa+\opa\sei\spk\sei\spj\big)\right),
\end{equation}
where $\hon,\htw,\htr\in\erre$ and $\kajk\in\ccc$, with $\kajk=\overline{\kakj}\fin$, $j,k\in\{1,2,3\}$
(note that the generator of a semigroup of positive linear maps must be adjoint-preserving).
Moreover, the semigroup is unital if and only if $\gene\sei\id=0$; i.e., specifically if and only if
$\kajk=\kakj\fin$. Therefore, for $\dim(\hh)=2$, the general form of the generator of a semigroup of unital, trace-preserving
and adjoint preserving linear maps is given by~{(\ref{geppo})}, with $\hon,\htw,\htr\in\erre$ and $\kajk=\kakj\in\erre$
(note that, since the matrix $(\kajk)$ is symmetric, the expressions~{(\ref{geppo})} and~{(\ref{geppo-bis})} are equivalent).
According to Proposition~{\ref{posetwo}}, the semigroup is also positive if and only if the real matrix $\rema$
(see~{(\ref{marep})} and~{(\ref{gefor})}), which by direct computation is given by
\begin{equation} \label{showrema}
\rema =
\begin{pmatrix}
-(\kat+\katr)/2 & \kaot/2-\htr & \kaotr/2+\htw \\
\kaot/2+\htr & -(\kao+\katr)/2 & \kattr/2-\hon \\
\kaotr/2-\htw &  \kattr/2+\hon & -(\kao+\kat)/2
\end{pmatrix} ,
\end{equation}
is such that $-(\rema + \remat)\ge 0$, where $-(\rema + \remat)$ is precisely the matrix~{(\ref{prema})}.
Finally, the linear map $\opa\mapsto\sum_{j,k=1}^3 \kajk\tre\spj\tre\opa\sei\spk$
--- equivalently, the semigroup of operators itself --- is completely positive if
and only if the symmetric real matrix $(\kajk)$ is positive semidefinite.
\end{proof}

\begin{remark}
By a well known classical result, a symmetric square matrix is positive semidefinite if and
only if all its principal minors are non-negative. Therefore, the matrix~{(\ref{prema})}
--- i.e., the matrix $\poma=-(\rema + \remat)$ --- is positive semidefinite if and only if
\begin{equation}
\ko, \kt , \ktr \ge 0 \fin , \
\ko\tre\kt \ge c^2 , \ \ko\tre\ktr \ge b^2 , \ \kt\tre\ktr \ge a^2 , \
\ko\tre\kt\tre\ktr + 2abc \ge \ko\tre a^2 + \kt\tre b^2 + \ktr\tre c^2  \fin .
\end{equation}
This set of inequalities is equivalent to the following:
\begin{equation}
\ko + \kt + \ktr \ge 0 \fin , \
\ko\tre\kt  +  \ko\tre\ktr  + \kt\tre\ktr \ge a^2 + b^2 + c^2 , \
\ko\tre\kt\tre\ktr + 2abc \ge \ko\tre a^2 + \kt\tre b^2 + \ktr\tre c^2  \fin .
\end{equation}
Of course, an analogous set of inequalities is necessary and sufficient for the symmetric real matrix
$\copoma\defi(\kajk)$ to be positive semidefinite.
\end{remark}

\begin{remark} \label{nopointed}
Note that the convex cone $\ptu$, $\dim(\hh)=2$, which is parametrized by the pair $\{\hbo,\poma\}$
--- with $\hbo=(\hon,\htw,\htr)\in\rcub$, $\poma$ positive semidefinite ---
is not `pointed' because
\begin{equation} \label{lineality}
\ptu \cap (-\ptu) = \left\{\ima\sei {\textstyle \sum_{j=1}^3} \hj \tre
\big[\spj,(\cdot)\big]\colon\ \hbo\in\rcub\right\};
\end{equation}
i.e., $\ptu$ is a `wedge'. Moreover, by inspecting the expression~{(\ref{showrema})},
it is clear that the matrix $\rema$ in~{(\ref{gefor})} is symmetric if and only if
the Hamiltonian component of the generator is zero ($\Leftrightarrow \gene=\genea$).
Clearly, the selfadjoint generators in $\ptu$ form a pointed convex cone.
\end{remark}

\begin{remark} \label{desgene}
Exploiting the fact that a symmetric real matrix can be diagonalized by a similarity transformation with
a matrix of $\rot$ and the covering homomorphism of $\sunig(2)$ onto $\rot$, one easily realizes that
the general form, for $\dim(\hh)=2$, of the generator of a semigroup of unital, trace-preserving, positive linear maps
is given by~{(\ref{geppo})} where $\hon,\htw,\htr\in\erre$ and $\copoma\defi(\kajk)=\diag(\gamma_1,\gamma_2,\gamma_3)$,
with $\gamma_1,\gamma_2,\gamma_3\in\erre$ satisfying~{(\ref{constra})}, and now the triple $\spo,\spt,\sptr$ is not fixed, but consists of
arbitrary selfadjoint operators verifying condition~{(\ref{absdef})}. Therefore, in this sense Example~{\ref{nonpos}}
(with condition~{(\ref{constra})}) turns out to describe the general form of a \emph{selfadjoint} qubit generator of
a semigroup of trace-preserving positive maps that does not decrease the entropy $\entro$.
\end{remark}

A further characterization, in dimension two, of the semigroups of
trace-preserving positive maps that do not decrease the entropy $\entro$ is the following.

\begin{proposition}
Let $\{\pds\colon\bH\rightarrow\bH\}_{t\in\errep}$ be a semigroup of
linear maps, with $\dim(\hh)=2$, and fix some $k\in\{1,2,3\}$ and
some complex conjugation $\coco$ in $\hh$. Then, the following facts are equivalent:
\begin{enumerate}

\item
$\{\pds\}_{t\in\errep}$ is positive, trace-preserving and unital;

\item
for some pair $(\repr,\{\sut\}_{t\in\errep})$ --- consisting of a Borel map $\repr\colon \lcs\rightarrow\unihh$,
where $\lcs$ is locally compact space, and of a family $\{\sut\}_{t\in\errep}$ of finite, signed Borel measures
on $\lcs$, with $\|\sut\|=\suta(\lcs)=1$, where $\suta$ is the total variation of $\sut$ --- the following
expression holds:
\begin{eqnarray}
\pds\sei\opa \spa & = & \spa \tr\big(\opa\big)\tre\spz +
2\sei\big\langle\spk,\opa\big\ranglehs
\int_{\lcs} \mtre \repr(x)\tre \spk \sei \repr(x)^\ast \; \de\sut(x)
\nonumber \\ \label{genpo}
& + & \spa
\motto\sum_{k\neq l\in\{1,2,3\}} \mdodici 2\sei\big\langle\spl,\opa\big\ranglehs
\int_{\lcs} \mtre \repr(x)\tre \spl \sei \repr(x)^\ast \; \de\suta(x) \fin ;
\end{eqnarray}

\item
for some triple $(\repr,\{\mut\}_{t\in\errep},\{\nut\}_{t\in\errep})$ --- consisting of a Borel map $\repr\colon \lcs\rightarrow\unihh$,
where $\lcs$ is locally compact space, and of two families $\{\mut\}_{t\in\errep}$, $\{\nut\}_{t\in\errep}$ of finite, positive Borel measures
on $\lcs$, with $(\mut +\nut)(\lcs)=1$ --- the following expression holds:
\begin{equation} \label{gentwise}
\pds\sei\opa = \int_{\lcs} \mtre \repr(x) \opa \sei \repr(x)^\ast \; \de\mut(x)
+ \int_{\lcs} \mtre \repr(x)\tre \big(\coco\tre\opa^\ast\mtre\coco\big) \tre \repr(x)^\ast \; \de\nut(x) \fin .
\end{equation}

\end{enumerate}
In particular, one can always set $\lcs=G\equiv\sunig(2)$ and choose $\repr\colon G\rightarrow\unihh$ as the defining
representation. Moreover, the semigroup of operators $\{\pds\}_{t\in\errep}$ is completely positive
(trace-preserving and unital) if and only if the finite measures $\{\sut\}_{t\in\errep}$ or $\{\mut\}_{t\in\errep}$
in~{(\ref{genpo})} or~{(\ref{gentwise})}, respectively, can be assumed to be probability measures,
and in this case $\{\pds\}_{t\in\errep}$ is a twirling semigroup.

Thus, a semigroup of trace-preserving positive maps $\{\pds\colon\bH\rightarrow\bH\}_{t\in\errep}$,
for $\dim(\hh)=2$, does not decrease the entropy $\entro$ if and only if it is of one of the equivalent
forms~{(\ref{genpo})} or~{(\ref{gentwise})}.
\end{proposition}

\begin{proof} According to a classical result~{\cite{Landau}}, every unital,
trace-preserving positive map in $\bH$, for $\dime=\dim(\hh)=2$, is of the form
\begin{equation} \label{utraprepo}
\randut\colon\bH\ni\opa\mapsto\sum_{j=1}^{\den} p_j\sei
\uj\mtre\big(\traspo^{\varepsilon(j)}\opa\big)\tre\ujast \fin ,
\end{equation}
where $\{p_j\}_{j=1}^{\den}\subset\erreps$ is a probability distribution,
$\big\{\uj\big\}_{j=1}^{\den}$ is a set of unitary operators in $\hh$,
$\traspo$ is a (fixed) transposition in $\bH$ and
$\varepsilon\colon\{1,\ldots,\den\}\rightarrow\{0,1\}$ is a map determining
its exponent ($\traspo^0\equiv\ide$). The map $\randut$ is completely positive
if and only if one can set $\varepsilon\equiv 0$; i.e., if and only if it is random unitary.
Let us fix the transposition $\traspo$ in such a way that, for some $k\in\{1,2,3\}$,
\begin{equation} \label{fixtra}
\traspo\tre\spk=-\spk \fin , \ \ \ \traspo\tre\spl=\spl \fin ,
\ \ k\neq l\in\{0,1,2,3\} \fin .
\end{equation}
Thus, we have that
\begin{equation}
\randut\big(\spk\big)= \sum_{j=1}^{\den} (-1)^{\varepsilon(j)} \tre p_j\sei
\uj\spk\sei\ujast \fin , \ \ \
\randut\big(\spl\big)= \sum_{j=1}^{\den} p_j\sei\uj\spl\sei\ujast
\fin , \ \ l\neq k \fin .
\end{equation}
Hence, for a suitable (real) linear combination $\su$ of Dirac measures on the group $G=\sunig(2)$
--- with $\|\sut\|=\suta(G)=1$ --- we have that
\begin{equation} \label{regi}
\randut\big(\spk\big)= \int_{G} \mtre \repr(g)\tre \spk \sei \repr(g)^\ast \; \de\su(g) \fin , \ \ \
\randut\big(\spl\big)= \int_{G} \mtre \repr(g) \tre\spl \sei \repr(g)^\ast \; \de\sua(g)
\fin , \ \ l\neq k \fin ,
\end{equation}
where $\repr\colon G\rightarrow\unihh$ is the defining representation.

Suppose now that, given a linear map $\randut\colon\bH\rightarrow\bH$,
relations~{(\ref{regi})} hold with the group $G$ replaced by some locally compact space
$\lcs$, and for some Borel map $\repr\colon \lcs\rightarrow\unihh$ and some finite, signed Borel measure
$\su$ on $\lcs$, such that $\|\su\|=\sua(\lcs)=1$. Let $\su=\supo-\sune$, with $\supo\perp\sune$
(i.e., the positive measures $\supo$ and $\sune$ are mutually singular),
be the Jordan decomposition of the signed measure $\su$,
so that $\sua=\supo+\sune$. We then have:
\begin{equation} \label{intsu}
\randut\big(\spj\big)= \int_{\lcs} \mtre \repr(x) \tre\spj \sei \repr(x)^\ast \; \de\supo(x)
+ \int_{\lcs} \mtre \repr(x)\tre \big(\traspo\tre\spj\big) \tre \repr(x)^\ast \; \de\sune(x) \fin ,
\end{equation}
where $\traspo$ is a transposition such that relations~{(\ref{fixtra})} hold. Since $\spz,\ldots,\sptr$
is a basis in $\bH$, setting $\mu\equiv\supo$, $\nu\equiv\sune$ (hence, $\mu(\lcs)+\nu(\lcs)=1$), we find:
\begin{equation} \label{intmunu}
\randut\tre\opa = \int_{\lcs} \mtre \repr(x) \opa \sei \repr(x)^\ast \; \de\mu(x)
+ \int_{\lcs} \mtre \repr(x)\tre \big(\coco\tre\opa^\ast\mtre\coco\big) \tre \repr(x)^\ast \; \de\nu(x) \fin ,
\ \ \ \opa\in\bH \fin ,
\end{equation}
where $\coco$ is a complex conjugation such that $\traspo\tre\opa=\coco\tre\opa^\ast\mtre\coco$.
Observe that every other complex conjugation in $\hh$ is of the form $\opu\coco=\coco\sei\opua$,
for some suitable $\opu\in\unihh$. Hence, using the fact that $\mu\perp\nu$,
one can easily check that~{(\ref{intmunu})} holds with $\coco$
replaced by any other complex conjugation, for a suitable redefinition of the map $\repr$ on a
$\mu$-negligible Borel subset $\bores$ of $\lcs$. Therefore, one gets a decomposition of $\randut$ of the form~{(\ref{intmunu})}
for \emph{every} fixed complex conjugation $\coco$. A similar observation holds true in the case where
$\lcs=G=\sunig(2)$ and $\repr\colon G\rightarrow\unihh$ is the defining representation;
but, in this case, one can leave the representation $\repr$ unchanged and replace the measure $\nu$
by a suitable right translate of this measure.

Assume, next, that $\randut$ is a linear map in $\bH$ such that a decomposition of the form~{(\ref{intmunu})}
holds, where now $\mu$, $\nu$ is a generic pair of finite positive Borel measures on $\lcs$ such that
$\mu(\lcs)+\nu(\lcs)=1$. Then, the map $\randut$ is positive, trace-preserving and unital.
Moreover, by a well known result in the theory of integration of vector-valued functions,
if $\mu\neq 0$ the first integral on the rhs of~{(\ref{intmunu})} belongs, up to a positive factor,
to the closed convex hull of the range of the integrand:
\begin{equation} \label{intraun}
\mu(\lcs)^{-1}\int_{\lcs} \mtre \repr(x) \opa \sei \repr(x)^\ast \; \de\mu(x)
\in \clco\big\{\opu\opa\sei\opua\colon\opu\in\repr(\lcs)\subset\unihh\big\} \fin .
\end{equation}
It follows that $\mu(\lcs)^{-1}\int_{\lcs} \mtre \repr(x)\tre (\cdot) \tre \repr(x)^\ast \; \de\mu(x)$ is
a random unitary map --- hence, completely positive --- and conversely, since $\dim(\hh)=2$, every
unital, trace-preserving, completely positive map in $\bH$ is random unitary~{\cite{Landau}};
hence, of the form indicated in~{(\ref{intraun})}, for some suitable pair  $(\repr,\mu)$, where $\mu$ is
a finite positive measure. Similarly, if $\nu\neq 0$ the second integral on the rhs of~{(\ref{intmunu})},
up to a positive factor (i.e., $\nu(\lcs)^{-1}$), amounts to applying the composition of a random unitary map
with a transposition. Thus, by decomposition~{(\ref{intmunu})} we see that $\randut$ is a convex combination
(with weights $\mu(\lcs)$ and $\nu(\lcs)$) of two unital trace-preserving maps:
the first one completely positive, the second one completely co-positive. Hence, it is a unital,
trace-preserving positive map.
In particular, $\randut$ is completely positive if and only if one can set $\nu=0$ in~{(\ref{intmunu})},
which amounts to assuming that $\mu$ (i.e., $\su$ if~{(\ref{intmunu})} is derived  from~{(\ref{intsu})})
is a probability measure.

By the previous arguments, one can easily check that the first property of the semigroup of
linear maps $\{\pds\}_{t\in\errep}$, listed in the statement to be proved, implies
the second, which implies the third. Then, the third property implies the first one,
and all other assertions are proven too.
\end{proof}

\section{Final remarks and conclusions}
\label{conclusions}

In this paper, we have considered the problem of characterizing a (finite-dimensional) quantum
evolution, governed by a quantum dynamical semigroup, that does not decrease a quantum entropy $\entro$,
for every initial state. The results that we have obtained hold for the von~Neumann, the Tsallis
and the R\'enyi entropies (the latter two, for any value of the parameter $q>0$), and for
a family of functions of density operators which are directly related to the Schatten norms.
All these quantities are not decreased, for every initial state, if and only if the
quantum dynamical semigroup is unital (Theorem~{\ref{mainth}}); i.e.,
if and only if the maximally mixed state $\mms$ is a stationary state
(thus the result does not depend on the choice of a particular type of entropy).
In terms of the infinitesimal generator $\gene$, this is equivalent to
the condition that the operators appearing in the Kraus-Stinespring-Sudarshan decomposition~{(\ref{kssfo})}
of the completely positive map that characterizes $\gene$ --- see~{(\ref{forgene})} ---
be jointly normal; i.e., satisfy relation~{(\ref{joinor})}. Banks, Susskind and Peskin~{\cite{Banks}}
found the selfadjointness of these operators to be a sufficient condition. We now know
that this condition amounts to selecting a special subclass (Example~{\ref{exgesa}})
out of the class of twirling semigroups (Example~{\ref{twisem}}), which is itself contained in the
whole class of unital quantum dynamical semigroups, and actually saturates it
for $\dime=\dim(\hh)=2$ (Corollary~{\ref{satur}}).

An interesting consequence of Theorem~{\ref{mainth}} is the fact that a quantum dynamical semigroup
enjoys the property of not decreasing any of the mentioned entropies if and only if the adjoint semigroup
--- w.r.t.\ the Hilbert-Schmidt scalar product --- is a quantum dynamical semigroup itself
(Corollary~{\ref{coradj}}). This property turns out to be also equivalent to the property of not
\emph{increasing} the purity (Corollary~{\ref{corpur}}).

Taking into account both Remark~{\ref{symschu}} and Remark~{\ref{symschu-bis}},
it is worth noting that the proof of the main result actually relies on the following properties
of the quantum entropies we have considered:
\begin{description}

\item{\tt (E1)}
the entropy $\entro\colon\stah\rightarrow\erre$ attains its (strict global) maximum value at $\mms$;

\item{\tt (E2)}
$\entro$ is Schur concave w.r.t.\ the natural majorization relation in $\stah$, namely,
\begin{equation} \label{schuconca-gen}
\ho \prec \hrho \ \defar \ \vepho\prec\veprho
\ \Longrightarrow \ \entro (\ho) \ge \entro (\hrho) \fin .
\end{equation}

\end{description}

It is clear, therefore, that the same results we have found will also hold
for any analogous quantity verifying the given `axioms'. E.g., if
$\acca$ is a real function on $\erren$, strictly decreasing (alternatively, strictly increasing)
w.r.t.\ each of its arguments, and $\go,\ldots,\gn$ are strictly convex
(respectively, strictly concave) continuous real functions on the interval $[0,1]$, then
\begin{equation} \label{genty}
\stah\ni\hrho\mapsto \acca(\tr(\go(\hrho)),\ldots,\tr(\gn(\hrho)))\ifed\entro(\hrho)
\end{equation}
satisfies the aforementioned axioms. Indeed --- arguing as in~{\cite{Marshall}}, chapter~3, C.1 ---
one shows that the map $\stah\ni\hrho\mapsto\tr(\gj(\hrho))=\sum_k \gj(\vepkrho)$, $j=1,\ldots,n$, is
strictly Schur convex (respectively, strictly Schur concave); so that, by the previous
assumptions on $\acca$, the function~{(\ref{genty})} is strictly Schur concave. Hence, it
verifies both axioms {\tt (E1)} and {\tt (E2)}. Note that the set of Schur concave functions of the general
type~{(\ref{genty})} contains all the generalized entropies considered, e.g., in~{\cite{Zozor}}
(in the context of the study of entropy-like uncertainty relations) and, in particular,
all the entropies considered in this paper, with the exception of $\neinf$
(which is, however, a point-wise limit of a family of functions of this kind, see~{(\ref{ainf})}).
For instance, in the case of the von~Neumann entropy, we have:
$n=1$, $\gi\colon [0,1]\ni \xi\mapsto -\xi\ln \xi$ (concave) and $\acca(x)=x$ (increasing);
while, for the Tsallis entropy $\tseq$, $q\neq 1$,
$\gi\colon [0,1]\ni \xi\mapsto \xi^q$ (concave for $q<1$, convex for $q>1$)
and $\acca(x)=(x-1)/(1-q)$ (increasing for $q<1$, decreasing for $q>1$).

In sect.~{\ref{further}}, we have investigated generalizations of the results
obtained in sect.~{\ref{main}} by considering the case of positive, but
not necessarily completely positive, dynamical semigroups. In this setting,
Theorem~{\ref{possem}} can be regarded as a generalization of Theorem~{\ref{mainth}},
and Proposition~{\ref{genertwi}} provides a characterization of the
positive dynamical semigroups that do not decrease a quantum entropy
by means of an integral expression which can be regarded as a generalization
of the formula defining a twirling semigroup. It seems therefore quite
natural to call such a semigroup of operators a \emph{generalized twirling semigroup}.
According to Corollary~{\ref{concon}}, the infinitesimal generators of the
generalized twirling semigroups acting in $\bH$ form a convex cone $\ptu$.
In the case where $\dim(\hh)=2$, the cone $\ptu$ can be described in detail;
see Proposition~{\ref{posetwo}}, Example~{\ref{nonpos}}, Proposition~{\ref{furpro}},
Remark~{\ref{nopointed}} and Remark~{\ref{desgene}}. Finally, Proposition~{\ref{posetwo}}
provides further integral expressions, for $\dim(\hh)=2$, of the semigroups of trace-preserving
positive maps.


\section*{Acknowledgments}

D.C.\ was partially supported by the National Science Center project 2015/17/B/ST2/02026.



\end{document}